%% file: m621.tex
\documentclass{ecai}

\usepackage{latexsym}
\usepackage{amssymb}
\usepackage{amsmath}
\usepackage{amsthm}
\usepackage{booktabs}
\usepackage{enumitem}
\usepackage{graphicx}
\usepackage{color}
\usepackage{hyperref}
\usepackage{subcaption}
\usepackage{bbm}
\usepackage{algorithm}
\usepackage[noend]{algorithmic}
\usepackage{xurl}
\usepackage{makecell}
\usepackage{multirow}
\usepackage[subtle]{savetrees}
\usepackage{bm}

\definecolor{myGreen}{RGB}{0, 128, 0}
\newtheorem{proposition}{Proposition}
\newtheorem{corollary}{Corollary}
\newtheorem{theorem}{Theorem}
\newtheorem{lemma}{Lemma}
\hypersetup{colorlinks=true, linkcolor=blue, citecolor=myGreen, urlcolor=cyan}
\DeclareMathOperator*{\argmax}{arg\,max}

\newcommand{\BibTeX}{B\kern-.05em{\sc i\kern-.025em b}\kern-.08em\TeX}

\begin{document}


\begin{frontmatter}


\paperid{621}


\title{Combining Diverse Information for Coordinated Action:\\Stochastic Bandit Algorithms for Heterogeneous Agents}


\author[A]{\fnms{Lucia}~\snm{Gordon}\orcid{0000-0003-3219-6960}\thanks{Corresponding Author. Email: \href{mailto:luciagordon@g.harvard.edu}{luciagordon@g.harvard.edu}}}
\author[A]{\fnms{Esther}~\snm{Rolf}\orcid{0000-0001-5066-8656}}
\author[A]{\fnms{Milind}~\snm{Tambe}\orcid{0000-0003-3296-3672}} 

\address[A]{Harvard University}


\begin{abstract}
Stochastic multi-agent multi-armed bandits typically assume that the rewards from each arm follow a fixed distribution, regardless of which agent pulls the arm. 
However, in many real-world settings, rewards can depend on the \emph{sensitivity} of each agent to their environment. 
In medical screening, disease detection rates can vary by test type; in preference matching, rewards can depend on user preferences; and in environmental sensing, observation quality can vary across sensors.
Since past work does not specify how to allocate agents of heterogeneous but known sensitivity of these types in a stochastic bandit setting, we introduce a UCB-style algorithm, \textsc{Min-Width}, which aggregates information from diverse agents. In doing so, we address the joint challenges of (i) aggregating the rewards, which follow different distributions for each agent-arm pair, and (ii) coordinating the assignments of agents to arms. 
\textsc{Min-Width} facilitates efficient collaboration among heterogeneous agents, exploiting the known structure in the agents' reward functions to weight their rewards accordingly. We analyze the regret of \textsc{Min-Width} and conduct pseudo-synthetic and fully synthetic experiments to study the performance of different levels of information sharing.
Our results confirm that the gains to modeling agent heterogeneity tend to be greater when the sensitivities are more varied across agents, while combining more information does not always improve performance.
\end{abstract}

\end{frontmatter}


\section{Introduction}\label{sec:intro}
The setting of stochastic multi-agent multi-armed bandits (MAB) \citep{multi-agent-network, social-learning, malicious-agents} is characterized by multiple agents taking actions simultaneously in each time step. This setting serves as a natural model for diverse domains, from COVID test allocation \citep{covid_greece} to preference matching \citep{stochastic-linear-bandits, multi-objective-user-utility} and poaching prevention \citep{dual-mandate, xu2022ranked}. These real-world problems involve unknown characteristics about the environment that are learned \emph{online} while the planner figures out the optimal action for each agent. The resulting explore-exploit tradeoff lends itself well to UCB-style algorithms \citep{Auer2002}, which estimate unknown quantities optimistically with an upper confidence bound (UCB) in an effort to maximize cumulative reward over time.

We introduce a new stochastic MAB problem wherein a planner specifies actions for agents of heterogeneous but known sensitivities to their unknown environment.
The ``environment'' comprises a set of arms, each of which takes on a state of 0 or 1 at each time step following a Bernoulli distribution, which models a binary outcome as in~\citet{solanki2023differentially} and \citet{dual-mandate}. The mean of the Bernoulli is an unknown parameter that must be learned online by the agents. The agents differ in their \textbf{sensitivity}, which is their probability of receiving a reward of 1 upon pulling an arm given that its state is 1. In this way, the utility of the agents' actions is a function of their sensitivity as well as the arm mean.

Several key ideas help us tackle the core challenges of sequential decision-making with multiple agents with heterogeneous sensitivities to their environment.
First, we address the combinatorial challenge of the many ways of allocating agents to arms by decomposing our combinatorial problem into learning the means of the individual arms. Second, we address the learning challenge by combining rewards across agents of varying sensitivity to speed up learning in a sensitivity-aware manner. Third, we address the problem of how to match \emph{heterogeneous} agents with arms by assigning the highest-sensitivity agents to the arms with the highest UCBs, which we experimentally show is an effective strategy. 
In contrast, applying past work to our problem without these insights would either naively combine all the agents' rewards and ignore their sensitivities \citep{CUCB} or slowly learn the optimal assignments by approaching the problem at the coarse, super-arm level \citep{Auer2002}.

We introduce the \textsc{Min-Width} algorithm designed for this new problem (\S\ref{sec:mw-algorithm}).
For each arm, \textsc{Min-Width} combines all the agents' rewards to generate a mean estimator with the tightest UCB, which is nontrivial since rewards are drawn from different distributions for each agent-arm pair. We derive an instance-independent $\mathcal{O}(\sqrt{T\log(T)})$ regret upper bound for the \textsc{Min-Width} algorithm, where $T$ is the time horizon and there are additional factors for the numbers of agents and their sensitivities (Theorem~\ref{theorem:min-width-regret}). We also evaluate \textsc{Min-Width} through pseudo-synthetic experiments with realistic parameter values for diverse domains including COVID test allocation, hotel recommendation, and poaching prevention along with fully synthetic experiments (\S\ref{sec:results}). To compare algorithms with different levels of information sharing, we introduce two sensitivity-aware baselines that we evaluate against \textsc{Min-Width}. We find that \textsc{Min-Width} outperforms classical baseline algorithms (CUCB \citep{CUCB} and UCB \citep{Auer2002}) not designed for heterogeneous agents as well as our sensitivity-aware baselines in many settings. Moreover, we show experimentally that the performance of \textsc{Min-Width} is robust to having only approximate knowledge of the agent sensitivities. 

\section{Motivating Domains}
\label{sec:application-domains}
Our setting of heterogeneous agents with known sensitivities is motivated by a diverse range of application domains, highlighted by the examples below. By explicitly incorporating agent heterogeneity and information-sharing, we introduce a more natural model for these domains.

\paragraph{COVID Test Allocation}
Consider the problem of allocating a limited number of COVID tests of varying sensitivity among floors of a college dorm with unknown virus prevalence to maximize detection of infected individuals. 
Different floors in a dorm, which serve as our ``arms,'' will likely have different prevalence rates of the virus.
The two primary types of tests to detect COVID are PCR (very sensitive) and antigen (less sensitive) tests, which serve as ``agents'' in our model.
Due to limited availability, suppose a college only has enough supply to test one student on each floor per time step. Thus, when we pull a super-arm, we distribute tests (the agents) among floors (the arms) and observe the test results.
\citet{covid_greece} consider a similar setup, developing a MAB system for airport COVID testing, but they do not account for varying test sensitivities.

\paragraph{Hotel Recommendation}
Consider the problem of matching customers with different preferences to hotels with limited space in order to maximize customer satisfaction, a task performed by websites such as Booking.com. Hotels, which are our ``arms,'' differ in features such as cleanliness, service, etc., which may not be known, especially for new hotels. Customers, which serve as our ``agents,'' vary in how much these features matter to them, information that booking platforms can request in a pre-recommendation survey. For instance, one customer may value cleanliness above other features and always be satisfied with their stay if the hotel was clean. Another customer may care equally about cleanliness and staff friendliness, so even if the hotel is clean they will not be satisfied if the staff are not very responsive. Assuming space in hotels is limited and some customer will have to be matched with a hotel that tends to be less clean, we maximize overall satisfaction by matching this latter customer with the less clean hotel and the former customer with the cleaner hotel on average.
A similar setting has been modeled using contextual bandits in \citet{multi-objective-user-utility}. However, unlike their algorithm we combine data from post-stay cleanliness reviews across customers with different preferences to better match either the same customers or new ones with similar preferences in future time steps.

\paragraph{Poaching Prevention}
Consider the problem of maximizing the detection of animal traps by planning patrols for rangers with different detection rates in a protected area such as a national park.
Poachers hide snares to trap animals, and rangers patrol the park for illegal activity and remove any snares they find, which can be modeled as a stochastic bandit problem \citep{dual-mandate}. Different parts of the park, which serve as our ``arms,'' are more or less likely to contain snares depending on their animal density, accessibility, etc., but the poaching rates across a park are often unknown due to the vastness of the areas. Rangers, the ``agents,'' vary in terms of their expertise, tools, and vehicles, giving them varying snare detection rates, or ``sensitivities,'' a type of real-world heterogeneity that has not been previously modeled.
Since the rangers are all patrolling the same park, we combine their observations to speed up our learning of poaching hotspots and optimize our assignments for rangers in the next round of patrols, accounting for the rangers' differing partial observabilities in snare detection.

\section{Related Work}
\label{sec:related_work}
\citet{Auer2002} introduce the UCB algorithm for the stochastic bandit problem that pulls the highest-UCB arm in every time step. \citet{variance-estimates} introduce the UCB-V algorithm that builds on UCB by incorporating arms' empirical variances. Neither of these model multiple agents. \citet{linear-CUCB} extend the UCB algorithm to the combinatorial setting where multiple arms are pulled in each time step with linear reward functions. \citet{CUCB} introduce an algorithm for possibly unknown reward functions. Both assume the reward obtained from pulling any given arm is an i.i.d. draw from a fixed distribution, an assumption that does \emph{not} hold in our problem. \citet{rejwan2019topk} consider the combinatorial bandit setting with full-bandit feedback (only the sum of the rewards is observed), whereas we operate in the semi-bandit feedback setting (the reward from each pull is observed) and also have heterogeneous agents.

Existing multi-agent bandit papers differ widely in their definition of agent heterogeneity. Some works consider agents with access to only a subset of the arms \citep{distributed_bandits}, differing but known communication abilities \citep{heterogeneous-stochastic-interactions}, or varying user preferences \citep{stochastic-linear-bandits, multi-objective-user-utility}. The latter is similar to our definition, but in their case the arm context is known and the user preferences are unknown, whereas we have unknown arm means but known sensitivities. Federated combinatorial bandits \citep{solanki2023differentially} also have heterogeneous agents, but their agents operate in a competitive environment and are subject to privacy constraints, whereas in our setting the agents are collaborating and there is no cost to their communication. 

Our notion of sensitivity is inspired by past work that utilizes sensor models to capture imperfect observability. In \citet{dual-mandate}, the agents do not observe the true state of each arm, though the more effort they exert, the more reliable their observations become. Their algorithm, however, assumes effort can be specified and distributed in each time step, whereas our agent sensitivities are fixed in advance. \citet{AdaSearch} model a single sensor trying to pick out the environment point with the strongest signal on average given observations that include contributions from all points but are more sensitive to those nearby. Unlike this work, we have multiple sensors, and we assume the rewards from distinct arms are independent.

Past work has explored the consequences of agent communication in different settings. In \citet{closing_the_gap}, multiple agents can be assigned to the same arm of a MAB, which results in a collision that can be used to transmit information between agents. We assume each agent pulls a distinct arm, so our agents cannot communicate in this way. In \citet{heterogeneous-stochastic-interactions}, the agents observe the actions and rewards of their neighbors with some known probability, where the agents are heterogeneous in terms of their ``sociability,'' so some may be more likely to observe their neighbors than others. Our agents' heterogeneity is unrelated to the way in which information is shared among them. \citet{me-in-team} consider a multi-agent explore-exploit optimization problem and demonstrate the uncertainty penalty phenomenon, wherein increased teamwork under uncertainty can degrade performance relative to the agents acting alone. Unlike their setting and the one in \citet{cesabianchi2022cooperative}, our problem has no intrinsic spatial nature.

\section{Problem Statement}
\label{sec:problem_statement}
We introduce a sequential decision-making problem in which a set of heterogeneous agents are allocated among a set of arms that yield stochastic rewards. There are $A$ agents $\mathcal{A}=\{a\}_{a=1}^A$, $N$ arms $\mathcal{N}=\{n\}_{n=1}^N$, and $T$ time steps $\mathcal{T}=\{t\}_{t=1}^T$.
The state of arm $n$ at time $t$ is a random variable $X_{t,n}\sim$ Bern$(\mu_n)$, so $X_{t,n}\in\{0,1\}$. The agents' heterogeneity is captured in their ``sensitivity,'' a scalar value associated with each agent. We denote the set of agent sensitivities by $\mathcal{S}=\{s_a\}_{a=1}^A$, where $s_a$ represents the probability agent $a$ receives a reward of 1 conditional on the true state of the arm being 1. Thus, the reward $Y_{t,a,n}$ obtained when pulling arm $n$ is a random variable that depends on both the agent $a$ who pulls it and the arm's mean:
\begin{equation}
    Y_{t,a,n}\sim \text{Bern}(s_a\mu_n).
    \label{eq:Y_tan}
\end{equation}
By construction, $s_a = \mathbb{P}[Y_{t,a,n} =1 | X_{t,n} = 1]$, and in our model, $\mathbb{P}[Y_{t,a,n} = 0 | X_{t,n} = 0] = 1$.

At each time step, the planner selects a super-arm assigning each agent to a distinct arm. The super-arm is chosen from the set $\mathcal{F}=\{f\}_{f=1}^\frac{N!}{(N-A)!}$, where $f:\mathcal{A}\to\mathcal{N}$ such that $f(a)\neq f(a')$ if $a\neq a'$. The super-arm selected at time $t$ is denoted $f_t$, and so $f_t(a)$ is the arm to which agent $a$ is assigned at time $t$.
We keep track of the number of times each arm has been pulled by each agent. Let
\begin{equation}
    \textstyle c_{t,a,n}=\sum_{\tau=1}^t\mathbbm{1}_{f_\tau(a)=n}
    \label{eq:c_tan}
\end{equation}
be the number of times arm $n$ has been pulled by agent $a$ through time $t$ and
\begin{equation}
     \textstyle c_{t,n}=\sum_{a=1}^Ac_{t,a,n}
    \label{eq:c_tn}
\end{equation}
be the total number of times arm $n$ has been pulled through time $t$ by any agent. Based on these, let us also define
\begin{equation}
        \mathcal{T}_{a,n}=\{t\in\mathcal{T}|c_{t,a,n}>0\}
        \label{eq:T_a,n}
\end{equation}
as the set of times agent $a$ has pulled arm $n$ at least once and 
\begin{equation}
        \mathcal{T}_n=\{t\in\mathcal{T}|c_{t,n}>0\}
        \label{eq:T_n}
\end{equation}
as the set of times for which arm $n$ has been pulled at least once by any agent. 

The total reward collected by pulling super-arm $f$ is a sum over the individual agent rewards: $r_f=\sum_{a=1}^AY_{t,a,f(a)}$. The expected reward is then
\begin{equation*}
    \textstyle \overline{r}_f=\mathbb{E}[r_f]=\sum_{a=1}^A\mathbb{E}[Y_{t,a,f(a)}]=\sum_{a=1}^As_a\mu_{f(a)}.
\end{equation*}
We define the optimal super-arm $f^\star$ to be the one that maximizes the expected reward:
\begin{equation}
     \textstyle f^\star=\argmax_{f\in \mathcal{F}}\overline{r}_f=\argmax_{f\in\mathcal{F}}\sum_{a=1}^As_a\mu_{f(a)}.
    \label{eq:f_star}
\end{equation}
The cumulative regret at time $T$ captures how poor the super-arms selected at the time steps elapsed so far perform compared to the optimal super-arm in expectation. In other words, we measure how much worse the cumulative expected reward is given a sequence of agent-arm assignments relative to the best it could be given the agent sensitivities at hand. The objective is to minimize the cumulative regret at time $T$, given by 
\begin{equation}
\textstyle  R_T=\sum_{t=1}^T\sum_{a=1}^As_a\left(\mu_{f^\star(a)}-\mu_{f_t(a)}\right).
    \label{eq:regret_def}
\end{equation}

\section{\textsc{Min-Width} Algorithm}
\label{sec:mw-algorithm}
\subsection{Algorithm Structure}
We introduce \textsc{Min-Width}, a UCB-style algorithm for assigning heterogeneous agents to stationary stochastic arms with Bernoulli rewards. We assume a centralized planner that knows the sensitivities of all the agents and coordinates their assignment to arms in each time step.
\textsc{Min-Width}, outlined in Algorithm~\ref{alg:algorithm}, revolves around an $N$-length vector of UCBs denoted by $\widetilde{\text{UCB}}$, where $\widetilde{\text{UCB}}_t$ represents the UCBs the planner uses to match each agent $a$ with an arm $n$ at time $t+1$. In each time step $1\leq t \leq T$, the agents are assigned to arms sequentially in descending order by sensitivity (lines~\ref{line:ranked-agents-for-loop}-\ref{line:remove-pulled-arm}) so that the highest-sensitivity agent is assigned to the arm with the highest $\widetilde{\text{UCB}}_{t-1}$, the next-highest-sensitivity agent is then assigned to the arm with the highest $\widetilde{\text{UCB}}_{t-1}$ out of those remaining unselected, and so on until all the agents have been assigned to distinct arms. The super-arm corresponding to this assignment is then pulled and each agent gets some reward $Y_{t,a,f_t(a)}$ (line~\ref{line:get-reward}). Next, the UCBs are updated $\widetilde{\text{UCB}}_{t-1}[n]\to\widetilde{\text{UCB}}_t[n]$ for every arm $n$ (line~\ref{line:update-rule}) according to Equation~\ref{eq:min-width-UCB_t,n}.

\subsection{Agent Allocation Strategy}
Our agent allocation strategy (lines~\ref{line:ranked-agents-for-loop}-\ref{line:remove-pulled-arm}) is inspired by the definition of the optimal super-arm in Equation~\ref{eq:f_star}: $f^\star$ assigns the $i$th-highest-sensitivity agent to the $i$th-highest-mean arm. In practice, we do not know the true means $\{\mu_n\}_{n\in\mathcal{N}}$. Instead, we can estimate them with some $\{\hat{\mu}_n\}_{n\in\mathcal{N}}$ and set some upper confidence bounds $\{\text{UCB}_n=\hat{\mu}_n+\epsilon_n\}_{n\in\mathcal{N}}$ on them, where $\epsilon_n$ is the width of the confidence interval around our estimate $\hat{\mu}_n$ of $\mu_n$. In the standard UCB algorithm, the optimal action is to pull the arm with the highest mean, and since that is unknown, the algorithm instead pulls the arm with the highest UCB \citep{Auer2002}. Analogously, in our setting, the optimal action is to match the $i$th-best agent with the arm with the $i$th-highest mean, but since the means are unknown, we instead match the $i$th-best agent with the arm with the $i$th-highest UCB. Another way to motivate this is to consider the infinite-data setting, in which the UCBs are equal to the true arm means. In that case, to maximize our expected reward we must match the highest-sensitivity agent to the highest-mean arm. In the finite-data setting, we optimistically estimate the means with the UCBs (which converge to the means with increasing amounts of data), which is why we match the highest-sensitivity agent with the highest-UCB arm.

\begin{algorithm}[t]
\caption{\textsc{Min-Width} Algorithm}
\label{alg:algorithm}
\begin{algorithmic}[1]
\STATE $\widetilde{\text{UCB}}_0\gets[\infty,...,\infty]$ (length $N$)
\STATE ranked\_agents $\gets$  flip(argsort(sensitivities))
\FOR{$t$ in range($1,T+1$)}
\STATE $f\gets[-1,...,-1]$ ($A$)
\STATE unassigned\_arms $\gets$ $(0\;,\;...\;,\;N-1)$
\FOR{$a$ in ranked\_agents}
\label{line:ranked-agents-for-loop}
\STATE UCBs$\gets\widetilde{\text{UCB}}_{t-1}[\text{unassigned\_arms}]$
\STATE $n\gets$ unassigned\_arms[argmax(UCBs)]
\STATE $f[a]\gets n$
\STATE unassigned\_arms.remove($n$)
\label{line:remove-pulled-arm}
\ENDFOR
\STATE pull super-arm $f$ and get $\{Y_{t,a,f_t(a)}\}_{a\in\mathcal{A}}$
\label{line:get-reward}
\FOR{$n$ in range(N)}
\STATE $\widetilde{\text{UCB}}_t[n]\gets$ update rule given by Equation~\ref{eq:min-width-UCB_t,n}
\label{line:update-rule}
\ENDFOR
\ENDFOR
\end{algorithmic}
\end{algorithm}

\subsection{Update Rule}
\label{subsubsec:min-width}
At each time step $t$, the planner uses all the rewards collected so far to generate a new $\text{UCB}_{t,n}$ for each arm $n$, as derived in \S\ref{sec:analysis}. In particular, the planner constructs an empirical estimator $\hat{\mu}_{t,n}$ for the mean of each arm that combines all the agents' rewards while accounting for their heterogeneity so as to minimize the width of the confidence interval $\epsilon_{t,n}$ around that estimator.
The empirical estimator of the mean of arm $n$ at time $t\in\mathcal{T}$ is given by
\begin{equation}
        \hat{\mu}_{t,n} = \begin{cases}
            t\notin\hyperref[eq:T_n]{\mathcal{T}_n} & 0.5 \\
            t\in\hyperref[eq:T_n]{\mathcal{T}_n} & \frac{\sum_{a=1}^A s_a \sum_{\tau=1}^t\mathbbm{1}_{f_\tau(a)=n} \hyperref[eq:Y_tan]{Y_{\tau,a,n}}}{\sum_{b=1}^A{s_b}^2\hyperref[eq:c_tan]{c_{t,b,n}}}.
        \end{cases}
        \label{eq:muhat_t,n}
    \end{equation}
The width of the confidence interval on $\mu_n$ at time $t\in\mathcal{T}$ is
\begin{equation}
        \epsilon_{t,n} = \sqrt{\frac{\ln(2NG(T,A)/\delta)}{2\sum_{a=1}^A {s_a}^2 \hyperref[eq:c_tan]{c_{t,a,n}}}}
        \label{eq:epsilon_t,n}
\end{equation}
for
\begin{equation}
    G(T,A)=\sum_{t=1}^T\binom{t+A-1}{A-1}<(T+1)^A.
        \label{eq:G(T,A)}
\end{equation}
The UCB on the mean of arm $n$ at time $t\in\mathcal{T}$ is
\begin{equation}
    \text{UCB}_{t,n}=\hyperref[eq:muhat_t,n]{\hat{\mu}_{t,n}} + \hyperref[eq:epsilon_t,n]{\epsilon_{t,n}}
    \label{eq:min-width-UCB_t,n}
\end{equation}
and is used as the $\widetilde{\text{UCB}}_t[n]$ in line \ref{line:update-rule}: $\widetilde{\text{UCB}}_t[n] = \text{UCB}_{t,n}$. Note that the algorithm returns a finite $\widetilde{\text{UCB}}_t[n]$ after a single pull of arm $n$ by any agent. The distribution of the rewards can vary greatly depending on the sensitivity of the agent who collected them, but because the planner knows all the agent sensitivities, they can harness that information to appropriately weight the rewards from different agents in generating a shared UCB.

\section{Theoretical Results}
\label{sec:analysis}
We provide analytical results for the \textsc{Min-Width} algorithm introduced in \S\ref{sec:mw-algorithm}, with complete proofs in Appendix~\ref{appendix:theory}. Incorporating all the agents' rewards to generate a shared UCB for each arm is a complex problem due to the agents' heterogeneity. Naively, one may think we could simply divide each reward for a given arm by the sensitivity of the agent who collected it and apply the original UCB algorithm to this sequence. This is invalid, however, because while these rescaled rewards will have identical means, they will still have different variances. The UCB algorithm assumes the rewards from a given arm are i.i.d. and hence would not apply to these rescaled rewards. To resolve this, in Proposition~\ref{prop:min-width-weights-derivation} we take a more general approach by treating this as an optimization problem where we optimize over weights on the agents' rewards to get the tightest confidence interval around the arm mean estimator.

\begin{proposition}(\textsc{Min-Width} Weights Derivation).
 \label{prop:min-width-weights-derivation}
    Suppose agent $a$ pulls arm $n$ a fixed number of times, a number we denote $c_{a,n}$, where the reward from each pull is $Y_{i,a,n}\sim\text{Bern}(s_a\mu_n)$. Let $\mathcal{C}_n=\{c_{a,n}\}_{a=1}^A$ contain the $c_{a,n}$ for every agent. Let $D_{\mathcal{C}_n,n}$ be the weighted sum of the independent rewards collected by all the agents from arm $n$, expressed in terms of weights $w_{\mathcal{C}_n,a,n}$:
\begin{equation}
    D_{\mathcal{C}_n,n}=\sum_{a=1}^Aw_{\mathcal{C}_n,a,n}\sum_{i=1}^{c_{a,n}}Y_{i,a,n}=\sum_{a=1}^A\sum_{i=1}^{c_{a,n}}w_{\mathcal{C}_n,a,n}Y_{i,a,n}.
    \label{eq:D_Cn,n}
\end{equation}
Then the weights $w_{\mathcal{C}_n,a,n}$ that minimize the width of the confidence interval on $\mu_n$ given by 
\begin{equation*}
    \textstyle \gamma_{\mathcal{C}_n,n}=\sqrt{\frac{\ln(2/\delta)}{2}\sum_{a=1}^A{w_{\mathcal{C}_n,a,n}}^2c_{a,n}}
\end{equation*}
under the constraint that the empirical estimator $D_{\mathcal{C}_n,n}$ is unbiased are
    \begin{equation}
        w_{\mathcal{C}_n,a,n}=\mathbbm{1}_{c_{a,n}>0}\times \frac{s_a}{\sum_{b=1}^A{s_b}^2c_{b,n}}.
        \label{eq:weights}
    \end{equation}
\end{proposition}

\begin{proof}
Since the rewards collected by a certain agent when pulling a certain arm are i.i.d., we consider weights on such sequences of i.i.d. rewards rather than on every single reward. If $c_{a,n}=0$, then agent $a$ has collected no rewards for arm $n$, and so we set $w_{\mathcal{C}_n,a,n}=0$ for any such agent. Hence, we need to solve for $w_{\mathcal{C}_n,a,n}$ only for agents with $c_{a,n}>0$. If $D_{\mathcal{C}_n,n}$ is to be unbiased, then we need $\mathbb{E}\left[D_{\mathcal{C}_n,n}\right]=\mu_n$, which sets the constraint 
$$\sum_{a=1}^Aw_{\mathcal{C}_n,a,n}s_ac_{a,n}=1.$$
Since $w_{\mathcal{C}_n,a,n}Y_{i,a,n}$ is bounded by $0\leq w_{\mathcal{C}_n,a,n}Y_{i,a,n}\leq w_{\mathcal{C}_n,a,n}$, Hoeffding's inequality gives
\begin{equation} \forall\delta\in(0,1),\;\mathbb{P}\left[\big|D_{\mathcal{C}_n,n}-\mu_n\big|\\<\sqrt{\frac{\ln(2/\delta)}{2}\sum_{a=1}^A{w_{\mathcal{C}_n,a,n}}^2c_{a,n}}\right]>1-\delta,
    \label{eq:min-width-Hoeffding-no-union}
\end{equation}
yielding $\gamma_{\mathcal{C}_n,n}=\sqrt{\frac{\ln(2/\delta)}{2}\sum_{a=1}^A{w_{\mathcal{C}_n,a,n}}^2c_{a,n}}$ as the width of the confidence interval on the mean of arm $n$ for some fixed number of pulls of each arm by each agent captured in $\mathcal{C}_n$.
To make this confidence interval as tight as possible, we solve for the weights that minimize 
$\gamma_{\mathcal{C}_n,n}$
under the constraint that $D_{\mathcal{C}_n,n}$ is unbiased for any non-random $\mathcal{C}_n$.
We solve this constrained optimization problem with the method of Lagrange multipliers, using Lagrangian
$$\mathcal{L}(w,\lambda)=\sqrt{\frac{\ln(2/\delta)}{2}\sum_{b=1}^A{w_{\mathcal{C}_n,b,n}}^2c_{b,n}}+\lambda\left(\sum_{b=1}^Aw_{\mathcal{C}_n,b,n}s_bc_{b,n}-1\right),$$
which results in 
$w_{\mathcal{C}_n,a,n}=\frac{s_a}{\sum_{b=1}^A{s_b}^2c_{b,n}}$, which holds for any agent $a$ with $c_{a,n}>0$. Since $w_{\mathcal{C}_n,a,n}=0$ for agents with $c_{a,n}=0$, we get Equation~\ref{eq:weights}.
\end{proof}

Next, in Theorem~\ref{theorem:min-width-concentration-bound}, we show that we can use the weights derived in Proposition~\ref{prop:min-width-weights-derivation} to construct an empirical estimator for the mean of each arm, whose deviation from the true mean we bound with high probability. This bound involves the challenge of counting the number of possible pulls of each arm by each agent.

\begin{theorem}(\textsc{Min-Width} Concentration Bound).
\label{theorem:min-width-concentration-bound}
    Suppose the empirical estimator of the mean of arm $n$ at time $t\in\mathcal{T}$ is given by Equation~\ref{eq:muhat_t,n}. Then for $\epsilon_{t,n}$ given in Equation~\ref{eq:epsilon_t,n}, $\hat{\mu}_{t,n}$ satisfies
    \begin{equation}
        \forall\delta\in(0,1),\;\mathbb{P}\left[\forall n\in\mathcal{N},\;\;t\in\mathcal{T},\;\;|\hyperref[eq:muhat_t,n]{\hat{\mu}_{t,n}}-\mu_n|<\hyperref[eq:epsilon_t,n]{\epsilon_{t,n}}\right]>1-\delta.
        \label{eq:mu^hat_tn_bound}
    \end{equation}
\end{theorem}

\begin{proof}
Applying a union bound over the arms to Equation~\ref{eq:min-width-Hoeffding-no-union} gives 
$\forall\delta\in(0,1)$,
\begin{equation*}
    \mathbb{P}\left[\forall n\in\mathcal{N},\;|D_{\mathcal{C}_n,n}-\mu_n|<\sqrt{\frac{\ln(\frac{2N}{\delta})}{2}\sum_{a=1}^A{w_{\mathcal{C}_n,a,n}}^2c_{a,n}}\right]>1-\delta.
\end{equation*}
Let $\mathcal{H}$ be the set of all possible instantiations of the set $\mathcal{C}_n$ assuming that arm $n$ has been pulled at least once within a time horizon of $T$, so $\mathcal{H}$ is a set of sets.
To apply a union bound over these sets, we determine the cardinality of $\mathcal{H}$ using the constraints that each element of $\mathcal{C}_n$ is between 0 and $T$ and the sum of the elements in $\mathcal{C}_n$ is between 1 and $T$. We denote the resulting cardinality $\hyperref[eq:G(T,A)]{G(T,A)}$, which would simply be $T$, as for CUCB, if all the agents were identical.
We perform the union bound, use that $\mathcal{C}_{t,n}\in\mathcal{H}\;\forall t\in\mathcal{T}_n$, and plug in for $\hyperref[eq:D_Cn,n]{D_{\mathcal{C}_{t,n},n}}$ and $\hyperref[eq:weights]{w_{\mathcal{C}_{t,n},a,n}}$, resulting in
$$\forall\delta\in(0,1),\;\;\mathbb{P}\left[\forall n\in\mathcal{N},\;\;t\in\mathcal{T}_n,\;\;|\hyperref[eq:muhat_t,n]{\hat{\mu}_{t,n}}-\mu_n|<\hyperref[eq:epsilon_t,n]{\epsilon_{t,n}}\right]>1-\delta.$$
For $t\notin\hyperref[eq:T_n]{\mathcal{T}_n}$, Equation~\ref{eq:muhat_t,n} gives $\hat{\mu}_{t,n}=0.5$, and $\epsilon_{t,n}=\infty$ since $c_{t,n}=0$. Equation~\ref{eq:mu^hat_tn_bound} follows directly since the difference between the true mean $\mu_n$ and 0.5 must be $<\infty$.
\end{proof}

Finally, in Theorem~\ref{theorem:min-width-regret} we use the concentration bound on the shared empirical mean from Theorem~\ref{theorem:min-width-concentration-bound} to upper bound the cumulative regret of the \textsc{Min-Width} algorithm.

\begin{theorem}(\textsc{Min-Width} Regret Bound).
    \label{theorem:min-width-regret}
    Suppose we act according to the \textsc{Min-Width} algorithm. Then $\forall\delta\in(0,1)$, the cumulative regret at time $T$ is bounded by
    \begin{equation}
         \mathbb{P}\left[R_T<A(N-1)\\+2\sqrt{2ANT\ln\left(\frac{2N\hyperref[eq:G(T,A)]{G(T,A)}}{\delta}\right)}\frac{\max\mathcal{S}}{\min\mathcal{S}}\right]>1-\delta.
        \label{eq:min_width_regret_bound}
    \end{equation}
\end{theorem}

\begin{proof}
    We bound $\mu_n-\hat{\mu}_{t,n}$ by $\mu_n-\hat{\mu}_{t,n}\leq|\hat{\mu}_{t,n}-\mu_n|$.
    Using the bound on $|\hat{\mu}_{t,n}-\mu_n|$ from Equation~\ref{eq:mu^hat_tn_bound} and the UCB on the mean of arm $n$ at time $t\in\mathcal{T}$ from Equation~\ref{eq:min-width-UCB_t,n} gives
    \begin{equation}
        \forall\delta\in(0,1),\;\;\mathbb{P}\left[\forall n\in\mathcal{N},\;\;t\in\mathcal{T},\;\;\mu_n<\text{UCB}_{t,n}\right]>1-\delta.
        \label{eq:mu_n_less_UCB}
    \end{equation}
    We split Equation~\ref{eq:regret_def} into terms with $t<N$ and $t\geq N$:
    \begin{multline}
        R_T=\sum_{t=1}^{N-1}\sum_{a=1}^As_a\left(\mu_{f^\star(a)}-\mu_{f_t(a)}\right)\\+\sum_{t=N}^T\sum_{a=1}^As_a\left(\mu_{f^\star(a)}-\mu_{f_t(a)}\right).
        \label{eq:R_T_split}
    \end{multline}
    Because $0<\mu_n<1$ $\forall n\in\mathcal{N}$, the difference between the means of any two arms is bounded by $\mu_n-\mu_{n'}<1\;\forall n,n'\in\mathcal{N}$. We apply this bound to the first term in Equation~\ref{eq:R_T_split} with $n=f^\star(a)$ and $n'=f_t(a)$ and also use the fact that $s_a\leq1$ $\forall a\in\mathcal{A}$, yielding
    \begin{equation}
        R_T<A(N-1)+\sum_{t=N}^T\sum_{a=1}^As_a\left(\mu_{f^\star(a)}-\mu_{f_t(a)}\right).
        \label{eq:regret_first_term_bounded}
    \end{equation}
    Let $R_{N:T}$ be the second term in Equation~\ref{eq:regret_first_term_bounded}. Using Equation~\ref{eq:mu_n_less_UCB} with $n=f^\star(a)$ to bound $\mu_{f^\star(a)}$ gives $\forall\delta\in(0,1)$,
    \begin{equation}
        \mathbb{P}\left[R_{N:T}<\sum_{t=N}^T\left(\sum_{a=1}^As_a\text{UCB}_{t,f^\star(a)}-\sum_{a=1}^As_a\mu_{f_t(a)}\right)\right]>1-\delta.
        \label{eq:regret_3}
    \end{equation}
    By construction, for all $t$ the \textsc{Min-Width} algorithm selects a configuration $f$ that maximizes $\sum_{a=1}^As_a\text{UCB}_{t,f_t(a)}$:
    \begin{equation}
        \forall t\in\mathcal{T},\;\;\sum_{a=1}^As_a\text{UCB}_{t,f^\star(a)}\leq\sum_{a=1}^As_a\text{UCB}_{t,f_t(a)}.
        \label{eq:using_alg}
    \end{equation}
    We use Equation~\ref{eq:using_alg} in Equation~\ref{eq:regret_3} and plug in Equation~\ref{eq:min-width-UCB_t,n} with $n=f_t(a)$, then use Equation~\ref{eq:mu^hat_tn_bound} and plug in for $\hyperref[eq:epsilon_t,n]{\epsilon_{t,f_t(a)}}$, yielding
    \begin{multline*}
        \forall\delta\in(0,1),\;\;\mathbb{P}\bigg[R_{N:T}<\sqrt{2\ln(2N\hyperref[eq:G(T,A)]{G(T,A)}/\delta)}\\\times\sum_{t=N}^T\sum_{a=1}^A\frac{s_a}{\sqrt{\sum_{b=1}^A {s_b}^2 \hyperref[eq:c_tan]{c_{t,b,f_t(a)}}}}\bigg]>1-\delta.
    \end{multline*}
    Note that $\forall a\in\mathcal{A},\;s_a\leq\max\;\mathcal{S}$ and $\forall b\in\mathcal{A},\;s_b\geq\min\;\mathcal{S}$. Using Equation~\ref{eq:c_tn} for $n=f_t(a)$ along with Lemma~\ref{appendix:lemma} yields 
    \begin{equation*}
        \forall\delta\in(0,1),\;\mathbb{P}\left[R_{N:T}<2\sqrt{2ANT\ln\left(\frac{2N\hyperref[eq:G(T,A)]{G(T,A)}}{\delta}\right)}\frac{\max\mathcal{S}}{\min\mathcal{S}}\right]>1-\delta.
    \end{equation*}
    Plugging this bound on $R_{N:T}$ into Equation~\ref{eq:regret_first_term_bounded} gives Equation~\ref{eq:min_width_regret_bound}, completing the proof. By definition of $\hyperref[eq:G(T,A)]{G(T,A)}$, the time dependence in $R_{N:T}$ is bounded above by $\mathcal{O}(\sqrt{T\ln(T)})$.
\end{proof}

\section{Experimental Setup}
\label{sec:experimental-setup}

We perform experiments to compare the efficacy of five algorithms: the one we design for this setting, \textsc{Min-Width}; two sensitivity-aware baselines we introduce, \textsc{No-Sharing} and \textsc{Min-UCB}; and two canonical baselines, \textsc{CUCB} and \textsc{UCB}.\footnote{The code and data are available on \href{https://github.com/lgordon99/heterogeneous-stochastic-bandits}{GitHub}~\cite{repo}.}

\subsection{\textsc{No-Sharing}}
\label{subsec:no-sharing}
The simplest information sharing setting is not to combine rewards across agents at all. In this \textsc{No-Sharing} strategy, each agent keeps track of their own UCB for each arm relying solely on their own rewards. The empirical estimator of the mean of arm $n$ according to agent $a$ with sensitivity $s_a$ at time $t\in\mathcal{T}$ is
\begin{equation}
        \hat{\mu}_{t,a,n}=\begin{cases}
            t\notin\hyperref[eq:T_a,n]{\mathcal{T}_{a,n}} & 0.5 \\
            t\in\hyperref[eq:T_a,n]{\mathcal{T}_{a,n}} & \frac{1}{s_a\hyperref[eq:c_tan]{c_{t,a,n}}}\sum_{\tau=1}^t\mathbbm{1}_{f_\tau(a)=n}\hyperref[eq:Y_tan]{Y_{\tau,a,n}}.
        \end{cases}
        \label{eq:muhat_t,a,n}
\end{equation}
Let the width of agent $a$'s confidence interval on the mean of arm $n$ at time $t\in\mathcal{T}$ for $\delta\in(0,1)$ be
\begin{equation}
        \epsilon_{t,a,n}=\frac{1}{s_a}\sqrt{\frac{\ln(2ANT/\delta)}{2\hyperref[eq:c_tan]{c_{t,a,n}}}}.
        \label{eq:epsilon_t,a,n}
\end{equation}
Agent $a$'s UCB on the mean of arm $n$ at time $t\in\mathcal{T}$ is then
\begin{equation}
    \text{UCB}_{t,a,n} = \hyperref[eq:muhat_t,a,n]{\hat{\mu}_{t,a,n}} + \hyperref[eq:epsilon_t,a,n]{\epsilon_{t,a,n}}.
    \label{eq:UCB_t,a,n}
\end{equation}
Here, each agent is almost operating in the standard UCB setting except for the assignment hierarchy, which has more sensitive agents pick which arms they want to pull before less sensitive agents. Consequently, less sensitive agents may have to pull arms that do not have the maximum $\text{UCB}_{t,a,n}$.

\subsection{\textsc{Min-UCB}}
\label{subsec:min-ucb}
The \textsc{Min-UCB} algorithm directly improves on the naive \textsc{No-Sharing} strategy. Each agent still keeps track of their own UCB for each arm, but since all the agent UCBs on the mean of a given arm hold simultaneously by Proposition~\ref{appendix:no-sharing}, we can take the minimum of these UCBs to get a tighter bound. The shared UCB for arm $n$ at time $t\in\mathcal{T}$ is then
\begin{equation}
    \text{UCB}_{t,n}=\min_{a\in\mathcal{A}}\;\hyperref[eq:UCB_t,a,n]{\text{UCB}_{t,a,n}}.
    \label{eq:min-UCB_UCB_t,n}
\end{equation}
In contrast to the \textsc{No-Sharing} algorithm, now agents effectively get information about arms they have not yet pulled since $\text{UCB}_{t,n}<\infty$ if \emph{any} agent has pulled arm $n$ even if agent $a$ has not.
The algorithm will match the $i$th-highest-sensitivity agent with the arm with the $i$th-highest $\text{UCB}_{t,n}$, still giving higher-sensitivity agents priority. 

While \textsc{Min-UCB} yields a tighter UCB than \textsc{No-Sharing}, it still ignores potentially valuable information by always using the UCB of one of the agents, which accounts for the rewards collected by that agent alone. If we want to tightly bound the mean on a given arm, intuitively it makes sense to use \emph{all} the rewards from the arm, not just those of whichever agent happens to have the lowest UCB for the arm. This is most evident in a setting where there are two agents of high sensitivity, such as 0.9 and 0.8. Perhaps the 0.9-agent has a lower UCB for an arm, but the pulls by the 0.8-agent represent additional rewards that could be used to further shrink the UCB that \textsc{Min-UCB} ignores, motivating our \textsc{Min-Width} algorithm that combines all the agents' rewards.

Note, however, that both \textsc{No-Sharing} and \textsc{Min-UCB} have $AT$ in the logarithm, which is smaller than \textsc{Min-Width}'s $\hyperref[eq:G(T,A)]{G(T,A)}$ factor. This may cause \textsc{Min-Width}'s UCBs to be higher than those of \textsc{No-Sharing} and \textsc{Min-UCB} in some cases. As a result, we anticipate that \textsc{Min-Width} may not always outperform \textsc{No-Sharing} and \textsc{Min-UCB}.

\subsection{\textsc{CUCB}}
\textsc{CUCB}~\citep{CUCB} combines all the rewards collected from each arm to generate a UCB for the arm. The algorithm is designed for sequences of i.i.d. rewards, which is not the case in our setting because of the agent heterogeneity. This algorithm has no way of accounting for heterogeneous agents, so we naively combine observations across agents in our implementation, ignoring the fact that the i.i.d. assumption does not hold. Consequently, it may never be able to learn the optimal agent-arm assignments. For \textsc{CUCB} we use Equation~\ref{eq:cucb_alg_ucb} for the UCB of arm $n$ at time $t$. Since this algorithm was not designed for heterogeneous agents, we randomly assign agents to the top-UCB arms at each time step.
\begin{multline}
    \text{UCB}_{t,n}=\frac{\sum_{\tau=1}^t\sum_{a=1}^A\mathbbm{1}_{f_\tau(a)=n}Y_{\tau,a,n}}{\sum_{\tau=1}^t\sum_{a=1}^A\mathbbm{1}_{f_\tau(a)=n}}\\+\sqrt{\frac{\ln(2Nt/\delta)}{2\sum_{\tau=1}^t\sum_{a=1}^A\mathbbm{1}_{f_\tau(a)=n}}}
    \label{eq:cucb_alg_ucb}
\end{multline}

\subsection{\textsc{UCB}}
The standard UCB algorithm \citep{Auer2002} maintains a UCB on the mean reward of every action that can be taken at each time step. In our implementation, we treat every super-arm as an arm and apply the UCB algorithm to every super-arm. We use Equation~\ref{eq:ucb_alg_ucb} for the UCB of super-arm $f$ at time $t$ and pull the super-arm with the highest UCB at each time step. By treating each super-arm as an arm, the \textsc{UCB} algorithm is implicitly able to account for the heterogeneity among the agents. However, it does not combine any information across agents or arms, making it increasingly unsuitable as the number of agents or arms increases.

\begin{multline}
    \text{UCB}_{t,f} = \frac{\sum_{\tau=1}^t\sum_{a=1}^A\mathbbm{1}_{f_\tau=f}Y_{\tau,a,f(a)}}{\sum_{\tau=1}^t\mathbbm{1}_{f_\tau=f}}\\+\sqrt{\frac{\ln(2N!t/\delta(N-A)!)}{2\sum_{\tau=1}^t\mathbbm{1}_{f_\tau=f}}}
    \label{eq:ucb_alg_ucb}
\end{multline}

\subsection{Implementation}
We implement \textsc{Min-Width}, \textsc{No-Sharing}, and \textsc{Min-UCB} as described in \S\ref{sec:mw-algorithm}, \S\ref{subsec:no-sharing}, and \S\ref{subsec:min-ucb}, respectively, setting $T\to t$ in Equations~\ref{eq:epsilon_t,n} and \ref{eq:epsilon_t,a,n}. While using $T$ facilitates the regret analysis, using $t$ in our experiments allows us to assess the performance at different times with a single run and also compare across simulations.
All graphs display the cumulative regret averaged over 90 trials with two standard errors and use $\delta=0.05$.

\section{Results}
\label{sec:results}
\begin{figure*}[ht]
     \centering
     \begin{subfigure}{0.35\textwidth}
         \centering
         \includegraphics[width=\textwidth]{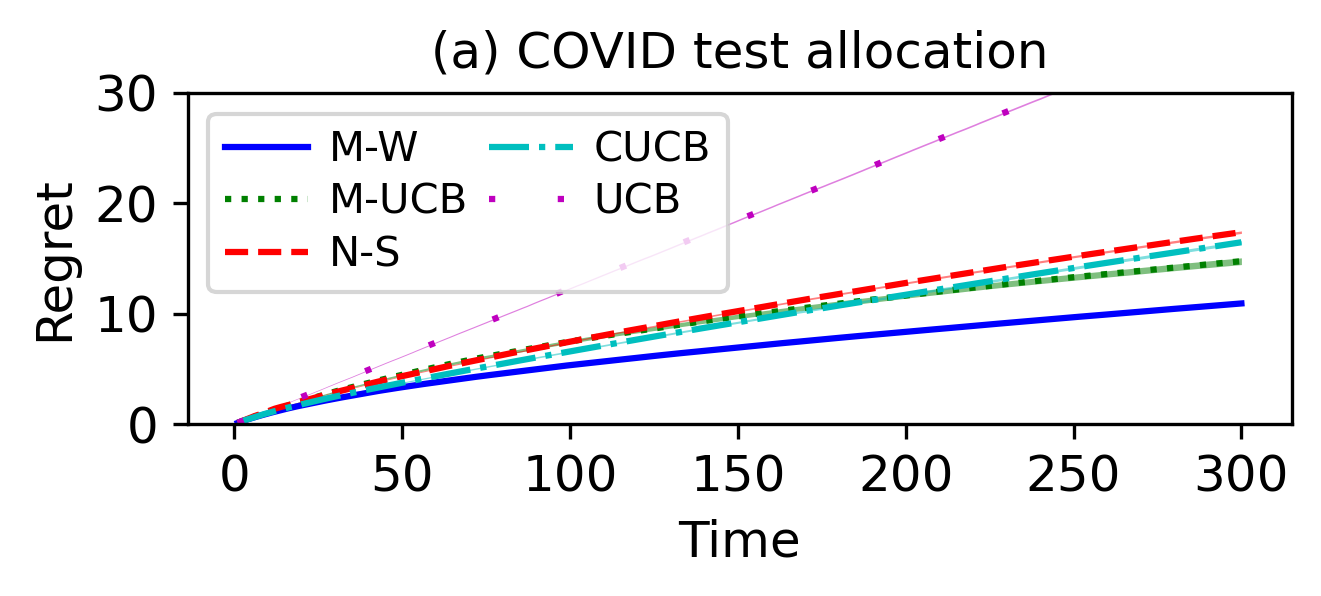}
     \end{subfigure}
     \begin{subfigure}{0.35\textwidth}
         \centering
         \includegraphics[width=\textwidth]{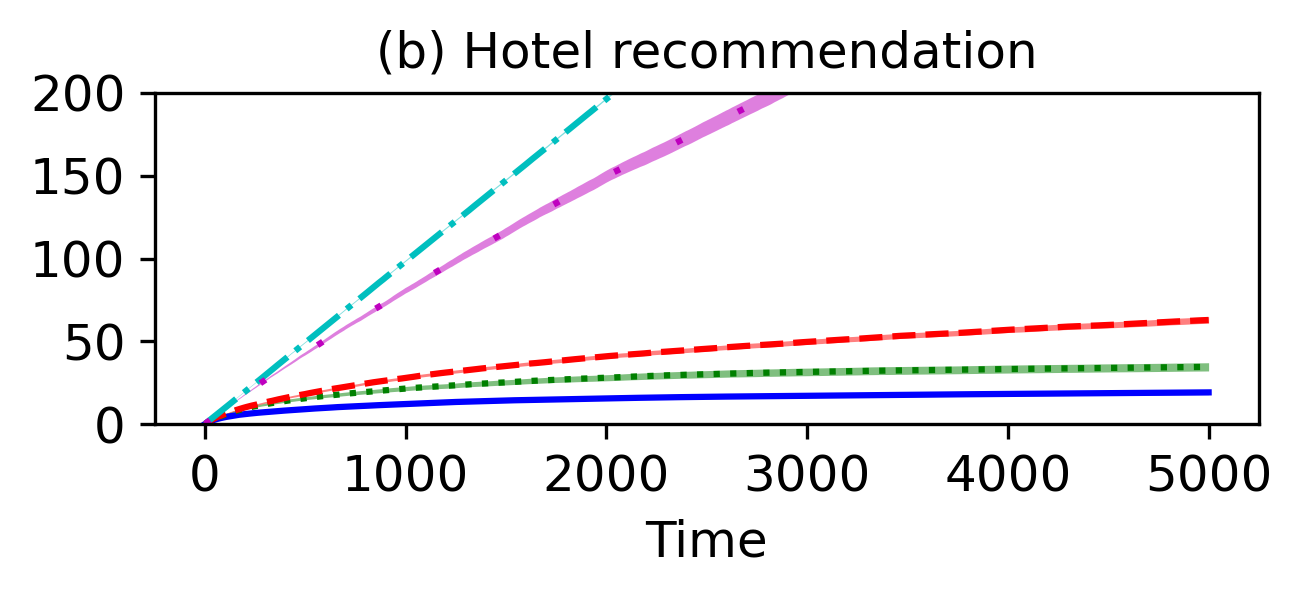}
     \end{subfigure}
        \caption{Regret plotted over time for the COVID test allocation (left) and hotel recommendation (right) domains.}
\label{fig:covid-hotel}
\end{figure*}
\begin{figure*}[ht]
     \centering
     \begin{subfigure}{0.33\textwidth}
         \centering
         \includegraphics[width=\textwidth]{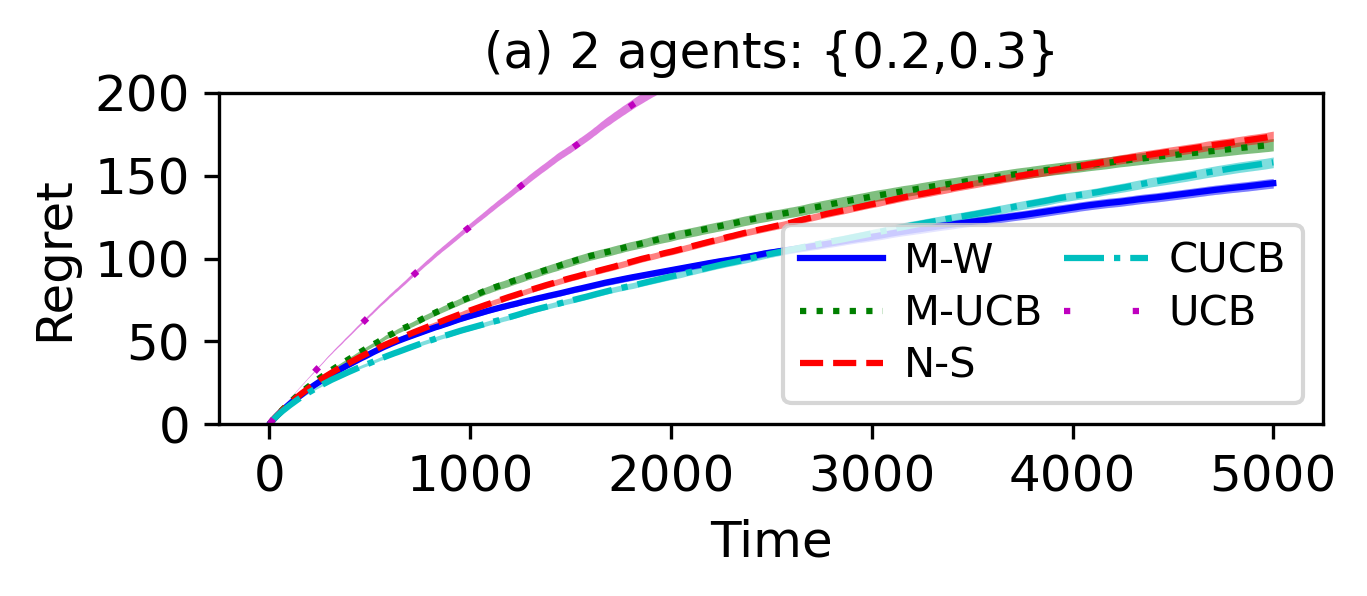}
     \end{subfigure}
     \hfill
     \begin{subfigure}{0.33\textwidth}
         \centering
         \includegraphics[width=0.95\textwidth]{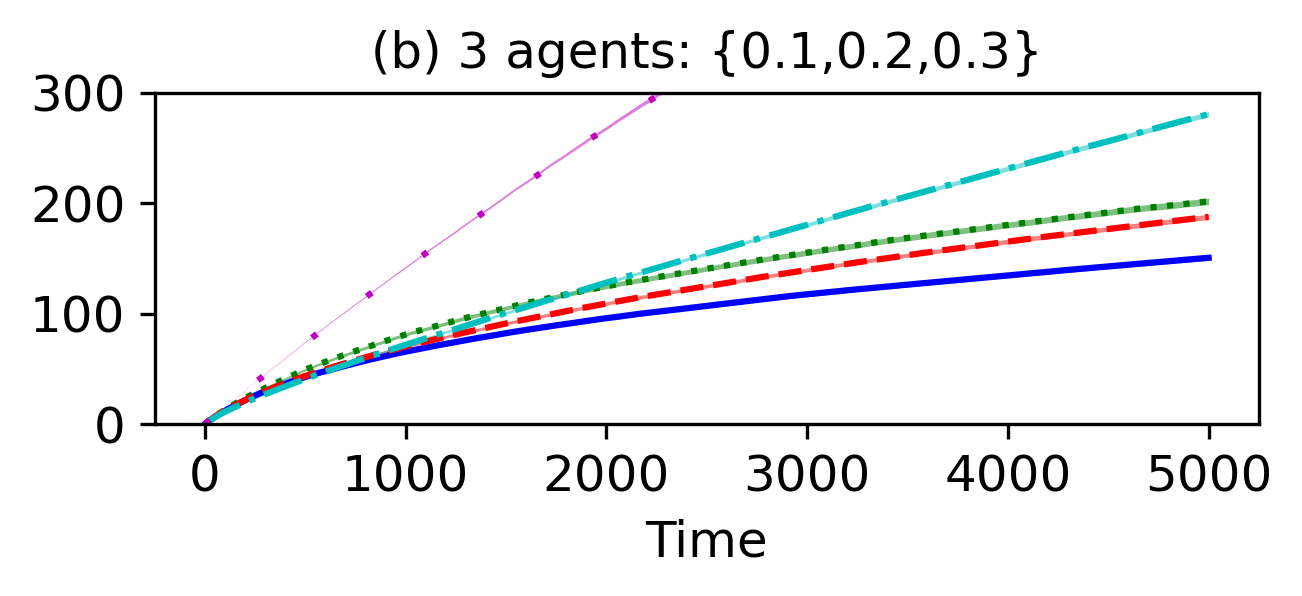}
     \end{subfigure}
     \hfill
     \begin{subfigure}{0.33\textwidth}
         \centering
         \includegraphics[width=0.95\textwidth]{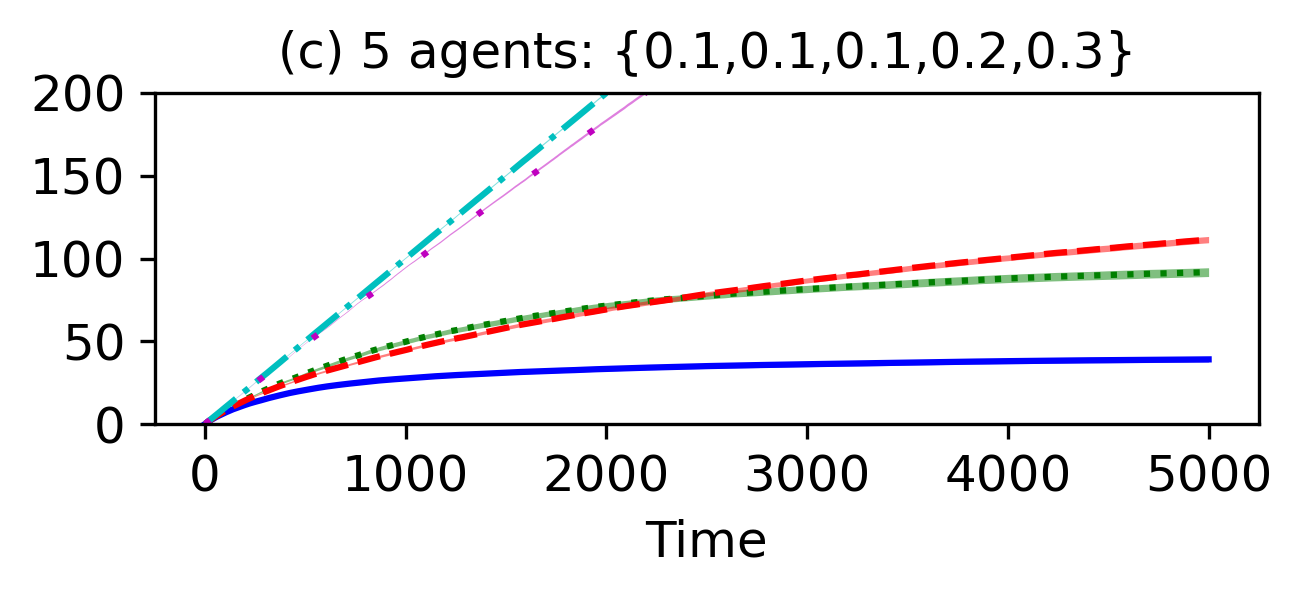}
     \end{subfigure}
        \caption{Regret plotted over time for the poaching prevention domain with varying agent sensitivities.}
    \label{fig:snaring-graph}
\end{figure*}

We perform simulations in four domains: three pseudo-synthetic domains with parameter values inspired by real data—COVID test allocation, hotel recommendation, and poaching prevention—and one fully synthetic domain, in which we vary the parameters of the simulations to study trends in the algorithms' behavior across problem settings. Finally, we test the robustness of the algorithms to estimating the agent sensitivities.

\subsection{COVID Test Allocation}
For our COVID simulation, suppose we have capacity to allocate COVID tests to 5 out of 6 floors of a college dorm with the following (assumed to be unknown) prevalence rates: $\mu=\{0.05,0.1,0.12,0.15,0.25,0.3\}$.
Each pull of a super-arm corresponds to distributing 5 tests among 6 floors, where the choice of who to test on each floor is random. We have 3 antigen tests and 2 PCR tests to distribute each day, with sensitivities to COVID of 80\% and 95\%, respectively \citep{covid-test-accuracies}: $\mathcal{S}=\{0.8,0.8,0.8,0.95,0.95\}$. We note that we are not making a prescriptive claim as to how schools should operate COVID testing but rather how our algorithm could be used to allocate tests accounting for the different sensitivities of the test types.

We simulate sequential super-arm pulls and show the results in Figure~\ref{fig:covid-hotel}a, where we see all the algorithms perform similarly for the initial time steps, after which \textsc{Min-Width} performs the best. \textsc{CUCB} performs similarly to \textsc{Min-UCB} and \textsc{No-Sharing}, which is reasonable since the test sensitivities are not very different, so ignoring their heterogeneity and combining all the results is not a terrible strategy.

\subsection{Hotel Recommendation}
For our hotel simulation, we extract cleanliness rates for four Punta Cana hotels with more than 500 reviews on TripAdvisor \citep{trip-advisor}.
The probability that each hotel is clean is the mean of these ratings:
$\mu=\{0.72,0.74,0.93,0.61\}$. Sensitivity is the probability that a customer is satisfied with their stay given that the hotel is clean. We consider a set of four customer types $\mathcal{S}=\{0.3,0.5,0.7,0.9\}$ and match each type with one of the four hotels, assuming that their relative cleanliness likelihoods are initially unknown. We aim to maximize overall satisfaction by matching the customers that care most about cleanliness with the cleanest hotels. Assuming that space at the hotels is limited and someone will need to be matched with a hotel that tends to be less clean, we incur the lowest cost to our customers' satisfaction if we match the customer whose satisfaction is least correlated with the hotel's cleanliness to a less clean hotel. That way, we keep spots open at the cleanest hotels for the customers who are very likely to be dissatisfied if their hotel is not clean.

In Figure~\ref{fig:covid-hotel}b, \textsc{Min-Width} outperforms the other four algorithms at all times, and the relative ordering of \textsc{Min-Width}, \textsc{Min-UCB}, and \textsc{No-Sharing} is consistent with the amount of information shared across agents. Both canonical baselines perform poorly; \textsc{UCB} has $4!=24$ super-arms to learn about, and the agent sensitivities vary widely, making \textsc{CUCB}'s assumption that the rewards from a given arm are i.i.d. even less appropriate than if the agent sensitivities were more similar.

\subsection{Poaching Prevention}
\label{subsec:poaching}

Consider a hypothetical park with five areas, each with a different probability of containing a snare: $\mu=\{0.1,0.3,0.5,0.7,0.9\}$. If rangers find a snare, they get a reward of 1, and if not, they receive 0 reward. Since even the best ranger teams probably cannot find more than $1/3$ of the snares present \citep{snare-detection}, we consider teams of two, three, and five rangers with sensitivities $\mathcal{S}=\{0.2,0.3\}$, $\mathcal{S}=\{0.1,0.2,0.3\}$, and $\mathcal{S}=\{0.1,0.1,0.1,0.2,0.3\}$, respectively. 

Comparing Figures~\ref{fig:snaring-graph}a-c shows that when there are fewer rangers, it takes much more time for the algorithms' performance to diverge, which makes sense because less information is collected during each round of patrols and the rangers' detection rates are similar. As we add rangers, even though they have very low sensitivity, the relative benefit of using \textsc{Min-Width} increases, as the low-sensitivity rangers assist with exploration and allow the higher-sensitivity rangers to exploit the best areas. The difference in performance between the two canonical baselines and the three sensitivity-aware algorithms also increases as we add rangers.

\subsection{Fully Synthetic}
We perform additional simulations for a variety of arm means and agent sensitivities to explore the long-range (high $T$) performance of the five algorithms, with results presented in Appendix~\ref{sec:additional_results}. The results, summarized in Table~\ref{table:synthetic-results}, demonstrate that this problem setting benefits from an approach distinct from a more general stochastic bandit algorithm.
In the $2\times2$ (2 agents, 2 arms) and $3\times3$ experiments \textsc{Min-Width} has the best long-range performance, which is consistent with Figures~\ref{fig:covid-hotel}b and \ref{fig:snaring-graph}c, in which \textsc{Min-Width} performed the best when the numbers of agents and arms were the same. When there are fewer agents than arms, \textsc{Min-UCB} sometimes outperforms \textsc{Min-Width}. This suggests that sharing rewards across agents is most useful when there are enough agents to continue exploring so that the best agents can exploit the seemingly best arms sooner.

\subsection{Sensitivity Robustness}
We explore how robust the sensitivity-aware algorithms' performance is to imperfect knowledge of the agent sensitivities, which we may not always know exactly.
When computing the UCBs and assigning agents to arms, the planner uses the estimated sensitivities rather than the true, assumed-to-be-unknown sensitivities used to compute the regret.
This change would have no effect on the canonical baselines, as they never model sensitivities explicitly.
We rerun the COVID experiment with three sets of estimated sensitivities: all overestimated ($\tilde{\mathcal{S}}=\{0.85,0.85,0.85,0.98,0.98\}$), all underestimated ($\tilde{\mathcal{S}}=\{0.75,0.75,0.75,0.9,0.9\}$), and a mix \\($\tilde{\mathcal{S}}=\{0.75,0.75,0.75,0.98,0.98\}$).
The percent change in regret is shown in Table~\ref{table:robustness} and the cumulative regret is shown in Table~\ref{table:robustness-regret} at $t=300$ across 500 trials.

For all algorithms, overestimating the sensitivities on average hurts performance less than underestimating them in this case. For \textsc{Min-Width} and \textsc{Min-UCB}, underestimating the antigen test sensitivities and overestimating the PCR test sensitivities hurts performance the most, possibly because the PCR tests are more frequently allocated to lower-COVID floors to perform exploration that the antigen tests are not considered accurate enough to reliably cover. Among the three sensitivity-aware algorithms, \textsc{No-Sharing} is least affected by these sensitivity approximations because the relative ordering of the agents is unchanged and the agents' UCBs are not affected by the estimation error in the sensitivities of the other agents. Overall, \textsc{Min-Width} had the lowest regret in the original experiment and is sufficiently robust for it to continue outperforming the other algorithms in the three robustness experiments (Figures~\ref{fig:overestimate}-\ref{fig:mix}).

\begin{table}
\caption{Percent change in regret (mean $\pm$ one SE) when using 
estimated vs. true sensitivities for the COVID experiment.}
\centering
\begin{tabular}{|c|c|c|c|}
\hline Algorithm & \makecell{Overestimated} & \makecell{Underestimated} & Mix \\
\hline \textbf{\textsc{M-W}} & $0.7\pm0.9$ & $2.5\pm1.0$ & $9.5\pm1.0$ \\
\hline \textsc{M-UCB} & $0.7\pm1.6$ & $9.1\pm1.8$ & $28.9\pm1.9$ \\
\hline \textsc{N-S} & $0.9\pm0.5$ & $1.5\pm0.5$ & $1.0\pm0.5$\\
\hline
\end{tabular}
\label{table:robustness}
\end{table}

\begin{table}
\caption{Cumulative regret (mean $\pm$ one SE) when using 
estimated vs. true sensitivities for the COVID experiment.}
\centering
\begin{tabular}{|c|c|c|c|}
\hline Algorithm & \makecell{Overestimated} & \makecell{Underestimated} & Mix \\
\hline \textbf{\textsc{M-W}} & $10.8\pm0.1$ & $11.0\pm0.1$ & $11.7\pm0.1$ \\
\hline \textsc{M-UCB} & $14.0\pm0.1$ & $15.2\pm0.2$ & $18.0\pm0.2$ \\
\hline \textsc{N-S} & $17.5\pm0.1$ & $17.6\pm0.1$ & $17.5\pm0.1$\\
\hline
\end{tabular}
\label{table:robustness-regret}
\end{table}

\section{Discussion}
Overall, either \textsc{Min-Width} or \textsc{Min-UCB} has the best long-range performance when the agents are heterogeneous, with \textsc{CUCB} performing best for identical agents. By combining observations across agents, its confidence intervals will shrink faster than those of \textsc{Min-UCB} and there is no agent heterogeneity it fails to account for in that case. When there are more arms than agents, \textsc{Min-UCB} sometimes achieves lower cumulative regret than \textsc{Min-Width}.
As remarked in \S\ref{subsec:min-ucb}, the confidence intervals for \textsc{Min-Width} may be wider because of the $\hyperref[eq:G(T,A)]{G(T,A)}$ factor from the union bound in Theorem~\ref{theorem:min-width-concentration-bound} over the possible instantiations of $C_n$, the set of times every agent has pulled arm $n$. This set depends on the number of agents but not their sensitivities, so it does not capture the fact that some agents are more similar than others, which we believe yields a resulting cardinality $\hyperref[eq:G(T,A)]{G(T,A)}$ that essentially overcounts.

For a clear example, consider two possible instantiations of the set $C_n$: \{4,5,6\} and \{5,4,6\}. In the first instantiation, agent 1 has pulled arm n four times, agent 2 five times, and agent 3 six times, while in the second instantiation, agent 1 has pulled it five times and agent 2 four times. Now suppose that agents 1 and 2 have the same sensitivity, making them interchangeable. In this case, these two instantiations of the set are not actually distinct and should only be counted as one instantiation rather than two (as is being done now). Future work could potentially resolve this by replacing the union bound with a different bound that is a function of the agent sensitivities.

Introducing the form of heterogeneous agent sensitivities that we study here gives rise to more realistic bandit models for myriad domains as discussed in \S\ref{sec:application-domains}. To model certain real-world applications with even more fidelity, future work could consider settings where the arm means change with time (e.g., COVID prevalence rates), possibly in response to arm pulls. One could also model agent sensitivities that vary across multiple dimensions (e.g., customer preferences for cleanliness, service, etc.). The constraint preventing multiple agents from pulling the same arm could also be relaxed, which would necessitate a model for interaction effects. Theoretically analyzing the estimated-sensitivity setting is another direction for future work that our robustness experiment opens up.

\section{Conclusion}
We introduce a stochastic multi-armed bandit problem with heterogeneous agents distinguished by a sensitivity parameter that uniquely characterizes their reward function for a given arm. We develop a method for assigning agents to arms that decomposes the combinatorial problem into one of learning the arm means while prioritizing the highest-sensitivity agents during arm assignment. Our \textsc{Min-Width} algorithm combines all of the rewards with a heterogeneity-aware weighting strategy. We provide a regret bound for \textsc{Min-Width} and evaluate it in simulations inspired by COVID test allocation, hotel recommendation, and poaching prevention, as well as a fully synthetic domain. Our results show that modeling agent heterogeneity tends to be most useful when the sensitivities are more diverse across a collection of agents and that sharing more information does not always improve performance.


\section*{Ethics Statement}
In our paper we introduce a model that may be able to more accurately capture certain elements of real-world systems than past work, especially those with varying agent sensitivities. Having a more complete model can lead to better performance in appropriate domains, which we demonstrate with our COVID test allocation, hotel recommendation, and poaching prevention experiments. Nonetheless, any automatic decision-making tool has inherent risks and should be used in appropriate contexts by people with knowledge of its capabilities and limitations. Any application of our methods to real-world settings would require much more perspective and sociotechnical analysis beyond the algorithmic contribution we offer here.


\begin{ack}
We would like to thank Lucas Janson for helpful discussions.
L.G. was supported by the National Science Foundation Graduate Research Fellowship under Grant No. DGE2140743.
E.R. was supported by the Harvard Data Science Initiative and the Harvard Center for Research on Computation and Society.
\end{ack}



\bibliography{m621}
\newpage
\onecolumn
\begin{center}
    \Large{Appendix}
\end{center}
\appendix
\input{appendix}
\end{document}

%% file: appendix.tex
\renewcommand{\thefigure}{\thesection\arabic{figure}}
\setcounter{figure}{0}
\renewcommand{\thetable}{\thesection\arabic{table}}
\setcounter{table}{0}

\section{Theory}
\label{appendix:theory}
\subsection{Definitions}
\begin{equation*}
    \mathcal{T}=\{t\}_{t=1}^T
\end{equation*}
\begin{equation*}
    Y_{t,a,n}\sim \text{Bern}(s_a\mu_n).
\end{equation*}
\begin{equation*}
    c_{t,a,n}=\sum_{\tau=1}^t\mathbbm{1}_{f_\tau(a)=n}
\end{equation*}
\begin{equation*}
        \mathcal{T}_{a,n}=\{t\in\mathcal{T}|c_{t,a,n}>0\}
\end{equation*}
\begin{equation*}
        \mathcal{T}_n=\{t\in\mathcal{T}|c_{t,n}>0\},
\end{equation*}
\begin{equation*}
        \hat{\mu}_{t,a,n}=\begin{cases}
            t\notin\mathcal{T}_{a,n} & 0.5 \\
            t\in\mathcal{T}_{a,n} & \frac{1}{s_ac_{t,a,n}}\sum_{\tau=1}^t\mathbbm{1}_{f_\tau(a)=n}Y_{\tau,a,n}
        \end{cases}
\end{equation*}
\begin{equation*}
        \epsilon_{t,a,n}=\frac{1}{s_a}\sqrt{\frac{\ln(2ANT/\delta)}{2c_{t,a,n}}}.
\end{equation*}
\begin{equation*}
    \text{UCB}_{t,a,n} = \hat{\mu}_{t,a,n} + \epsilon_{t,a,n}
\end{equation*}
\begin{equation*}
    \text{UCB}_{t,n}=\min_{a\in\mathcal{A}}\text{UCB}_{t,a,n}.
\end{equation*}
\begin{equation*}
    \text{UCB}_{t,n}=\hat{\mu}_{t,n} + \epsilon_{t,n}
\end{equation*}
\begin{equation*}
    R_T=\sum_{t=1}^T\sum_{a=1}^As_a\left(\mu_{f^\star(a)}-\mu_{f_t(a)}\right).
\end{equation*}

\subsection{\textsc{No-Sharing} Concentration Bound}
\label{appendix:no-sharing}
\begin{proposition}[\textsc{No-Sharing} Concentration Bound]
    Suppose the empirical estimator of the mean of arm $n$ according to agent $a$ with sensitivity $s_a$ at time $t\in\mathcal{T}$ is given by Equation~\ref{eq:muhat_t,a,n},
    \begin{equation*}
        \hat{\mu}_{t,a,n}=\begin{cases}
            t\notin\mathcal{T}_{a,n} & 0.5 \\
            t\in\mathcal{T}_{a,n} & \frac{1}{s_ac_{t,a,n}}\sum_{\tau=1}^t\mathbbm{1}_{f_\tau(a)=n}Y_{\tau,a,n},
        \end{cases}
    \end{equation*}
    for $Y_{\tau,a,n}$, $c_{t,a,n}$, and $\mathcal{T}_{a,n}$ defined in Equations~\ref{eq:Y_tan}, \ref{eq:c_tan}, and \ref{eq:T_a,n}, respectively. Then $\hat{\mu}_{t,a,n}$ satisfies
    \begin{equation}
        \forall \delta \in (0,1), \;\mathbb{P}\left[\forall a\in\mathcal{A},\;n\in\mathcal{N},\;t\in\mathcal{T},\;\left| \hat{\mu}_{t,a,n} - \mu_n \right| < \epsilon_{t,a,n} \right] \\ > 1 - \delta
        \label{eq:mu^hat_tan_bound}
    \end{equation}
    for $\epsilon_{t,a,n}$ as given in Equation~\ref{eq:epsilon_t,a,n},
    \begin{equation*}
        \epsilon_{t,a,n}=\frac{1}{s_a}\sqrt{\frac{\ln(2ANT/\delta)}{2c_{t,a,n}}}.
    \end{equation*}
\end{proposition}

\begin{proof}
Let $D_{c,a,n}$ be the sum of $c$ i.i.d. rewards \\$Y_{i,a,n}\sim\text{Bern}(s_a\mu_n)$, where $c$ is fixed and not random.
\begin{equation*}
    D_{c,a,n} = \sum_{i=1}^c Y_{i,a,n}
\end{equation*}
Note that each $Y_{i,a,n}$ is bounded by
\begin{equation*}
    0 \leq Y_{i,a,n} \leq 1.
\end{equation*}
Hoeffding's inequality then gives
\begin{equation*}
    \forall \alpha > 0, \mathbb{P}\left[| D_{c,a,n}-\mathbb{E}[D_{c,a,n}] | \geq \alpha \right] \\ \leq 2 \exp\left(-\frac{2\alpha^2}{\sum_{i=1}^c (1-0)^2}\right).
\end{equation*}
Plugging in $D_{c,a,n}$ yields
\begin{equation*}
    \forall \alpha > 0, \mathbb{P}\left[\left| \sum_{i=1}^c Y_{i,a,n} - \mathbb{E}\left[\sum_{i=1}^c Y_{i,a,n}\right] \right| \geq \alpha \right] \\ \leq 2 \exp\left(-\frac{2\alpha^2}{\sum_{i=1}^c 1}\right).
\end{equation*}
By the linearity of expectation, we get
\begin{equation*}
    \forall \alpha > 0, \mathbb{P}\left[\left| \sum_{i=1}^c Y_{i,a,n} - \sum_{i=1}^c \mathbb{E}[Y_{i,a,n}] \right| \geq \alpha \right] \leq 2 \exp\left(-\frac{2\alpha^2}{c}\right).
\end{equation*}
Plugging in the expectation of $Y_{i,a,n}$ gives
\begin{equation*}
    \forall \alpha > 0, \mathbb{P}\left[\left| \sum_{i=1}^c Y_{i,a,n} - \sum_{i=1}^c s_a \mu_n \right| \geq \alpha \right] \leq 2 \exp\left(-\frac{2\alpha^2}{c}\right).
\end{equation*}
Simplifying yields
\begin{equation*}
    \forall \alpha > 0, \mathbb{P}\left[\left| \sum_{i=1}^c Y_{i,a,n} - s_a c \mu_n \right| \geq \alpha \right] \leq 2 \exp\left(-\frac{2\alpha^2}{c}\right).
\end{equation*}
Taking the complement of the equation results in
\begin{equation*}
    \forall \alpha > 0, \mathbb{P}\left[\left| \sum_{i=1}^c Y_{i,a,n} - s_a c \mu_n \right| < \alpha \right] > 1 - 2 \exp\left(-\frac{2\alpha^2}{c}\right).
\end{equation*}
Set
\begin{equation*}
    \delta = 2 \exp\left(-\frac{2\alpha^2}{c}\right).
\end{equation*}
We can now solve for $\alpha$ in terms of $\delta$.
\begin{equation*}
    2 \exp\left(-\frac{2\alpha^2}{c}\right) = \delta \Longrightarrow \exp\left(-\frac{2\alpha^2}{c}\right) = \frac{\delta}{2} \Longrightarrow -\frac{2\alpha^2}{c}=\ln\left(\frac{\delta}{2}\right)\\\Longrightarrow\frac{2\alpha^2}{c}=\ln\left(\frac{2}{\delta}\right)\Longrightarrow\alpha=\sqrt{\frac{c\ln(2/\delta)}{2}}
\end{equation*}
The inequality then becomes
\begin{equation*}
    \forall \delta \in (0,1),\; \mathbb{P}\left[\left| \sum_{i=1}^c Y_{i,a,n} - s_a c \mu_n \right| < \sqrt{\frac{c\ln(2/\delta)}{2}} \right] > 1 - \delta.
\end{equation*}
Applying a union bound over
\\$\{\left| \sum_{i=1}^c Y_{i,a,n} - s_a c \mu_n \right| < \sqrt{c\ln(2/\delta)/2}\}_{a\in\mathcal{A}}$ gives
\begin{equation*}
    \forall \delta \in (0,1), \;\mathbb{P}\left[\forall a\in\mathcal{A},\left| \sum_{i=1}^c Y_{i,a,n} - s_a c \mu_n \right| < \sqrt{\frac{c\ln(2A/\delta)}{2}} \right] \\> 1 - \delta.
\end{equation*}
Applying a union bound over
\\$\{\left| \sum_{i=1}^c Y_{i,a,n} - s_a c \mu_n \right| < \sqrt{c\ln(2A/\delta)/2}\}_{n\in\mathcal{N}}$ gives
\begin{equation*}
    \forall \delta \in (0,1), \;\mathbb{P}\biggr[\forall a\in\mathcal{A},n\in\mathcal{N},\left| \sum_{i=1}^c Y_{i,a,n} - s_a c \mu_n \right| \\ < \sqrt{\frac{c\ln(2AN/\delta)}{2}} \biggr] > 1 - \delta.
\end{equation*}
Applying a union bound over
\\$\{\left| \sum_{i=1}^c Y_{i,a,n} - s_a c \mu_n \right| < \sqrt{c\ln(2AN/\delta)/2}\}_{c\in\mathcal{T}}$ gives
\begin{equation*}
    \forall \delta \in (0,1), \;\mathbb{P}\biggr[\forall a\in\mathcal{A},n\in\mathcal{N},c\in\mathcal{T},\left| \sum_{i=1}^c Y_{i,a,n} - s_a c \mu_n \right| \\ < \sqrt{\frac{c\ln(2ANT/\delta)}{2}} \biggr] > 1 - \delta.
\end{equation*}
Even though $c_{t,a,n}$ is a random variable, since $c_{t,a,n}\in\mathcal{T}\;\forall t\in\mathcal{T}_{a,n}$ for $\mathcal{T}_{a,n}$ defined in Equation~\ref{eq:T_a,n}, it holds that
\begin{equation*}
    \forall \delta \in (0,1), \;\mathbb{P}\biggr[\forall a\in\mathcal{A},n\in\mathcal{N},t\in\mathcal{T}_{a,n},\left| \sum_{i=1}^{c_{t,a,n}} Y_{i,a,n} - s_a c_{t,a,n} \mu_n \right| \\ < \sqrt{\frac{c_{t,a,n} \ln(2ANT/\delta)}{2}} \biggr] > 1 - \delta.
\end{equation*}
Dividing both sides of the inequality by $s_ac_{t,a,n}$ gives
\begin{equation*}
    \forall \delta \in (0,1), \;\mathbb{P}\biggr[\forall a\in\mathcal{A},n\in\mathcal{N},t\in\mathcal{T}_{a,n},\left| \frac{1}{s_ac_{t,a,n}} \sum_{i=1}^{c_{t,a,n}} Y_{i,a,n} - \mu_n \right| \\ < \frac{1}{s_a} \sqrt{\frac{ \ln(2ANT/\delta)}{2c_{t,a,n}}} \biggr] > 1 - \delta.
\end{equation*}
By the definition of $c_{t,a,n}$ in Equation~\ref{eq:c_tan}, summing over $i$ from 1 to $c_{t,a,n}$ is equivalent to summing over $\tau$ from 1 to $t$ and multiplying by the indicator variable $\mathbbm{1}_{f_\tau(a)=n}$.
\begin{equation*}
    \forall \delta \in (0,1), \;\mathbb{P}\biggr[\forall a\in\mathcal{A},n\in\mathcal{N},t\in\mathcal{T}_{a,n},\\\left| \frac{1}{s_ac_{t,a,n}} \sum_{\tau=1}^t \mathbbm{1}_{f_\tau(a)=n}Y_{\tau,a,n} - \mu_n \right| < \frac{1}{s_a} \sqrt{\frac{ \ln(2ANT/\delta)}{2c_{t,a,n}}} \biggr] > 1 - \delta
\end{equation*}
Using the definitions of $\hat{\mu}_{t,a,n}$ and $\epsilon_{t,a,n}$ given in Equations~\ref{eq:muhat_t,a,n} and \ref{eq:epsilon_t,a,n}, respectively, gives
\begin{equation*}
    \forall \delta \in (0,1), \;\mathbb{P}\left[\forall a\in\mathcal{A},n\in\mathcal{N},t\in\mathcal{T}_{a,n},\left| \hat{\mu}_{t,a,n} - \mu_n \right| < \epsilon_{t,a,n} \right] \\ > 1 - \delta.
\end{equation*}
Recall that for $t\notin\mathcal{T}_{a,n}$, we have $\hat{\mu}_{t,a,n}=0.5$ by Equation~\ref{eq:muhat_t,a,n} and $\epsilon_{t,a,n}=\infty$ since $c_{t,a,n}=0$. Because the difference between the true mean $\mu_n$ and 0.5 is less than infinity with probability 1, we get
\begin{equation*}
    \forall \delta \in (0,1), \;\mathbb{P}\left[\forall a\in\mathcal{A},n\in\mathcal{N},t\in\mathcal{T},\left| \hat{\mu}_{t,a,n} - \mu_n \right| < \epsilon_{t,a,n} \right] \\ > 1 - \delta.
\end{equation*}
\end{proof}

\subsection{\textsc{Min-UCB} Concentration Bound}
\label{appendix:min-UCB}
\begin{corollary}[\textsc{Min-UCB} Concentration Bound]
    Suppose the shared UCB on $\mu_n$ at time $t\in\mathcal{T}$ is given by Equation~\ref{eq:min-UCB_UCB_t,n},
    \begin{equation*}
        \text{UCB}_{t,n}=\min_{a\in\mathcal{A}}\text{UCB}_{t,a,n},
    \end{equation*}
    for $\text{UCB}_{t,a,n}$ defined in Equation~\ref{eq:UCB_t,a,n}. Then $\mu_n$ is bounded with high probability according to 
    \begin{equation}
        \forall\delta\in(0,1),\;\mathbb{P}\left[\forall n\in\mathcal{N},t\in\mathcal{T},\mu_n<\text{UCB}_{t,n}\right]>1-\delta.
        \label{eq:mu_n_bound}
    \end{equation}
\end{corollary}

\begin{proof}
    First, we can bound $\mu_n-\hat{\mu}_{t,a,n}$ by
    \begin{equation*}
        \mu_n-\hat{\mu}_{t,a,n}\leq|\mu_n-\hat{\mu}_{t,a,n}|=|\hat{\mu}_{t,a,n}-\mu_n|.
    \end{equation*}
    Using the bound on $|\hat{\mu}_{t,a,n}-\mu_n|$ from Equation~\ref{eq:mu^hat_tan_bound} gives
    \begin{equation*}
        \forall \delta \in (0,1), \;\mathbb{P}\left[\forall a\in\mathcal{A},n\in\mathcal{N},t\in\mathcal{T},\mu_n-\hat{\mu}_{t,a,n}< \epsilon_{t,a,n} \right] > 1 - \delta.
    \end{equation*}
    This is equivalent to
    \begin{equation*}
        \forall \delta \in (0,1), \;\mathbb{P}\left[\forall a\in\mathcal{A},n\in\mathcal{N},t\in\mathcal{T},\mu_n<\hat{\mu}_{t,a,n} + \epsilon_{t,a,n}\right] > 1 - \delta.
    \end{equation*}
    Using the definition of $\text{UCB}_{t,a,n}$ in Equation~\ref{eq:UCB_t,a,n} gives
    \begin{equation*}
        \forall \delta \in (0,1), \;\mathbb{P}\left[\forall a\in\mathcal{A},n\in\mathcal{N},t\in\mathcal{T},\mu_n<\text{UCB}_{t,a,n}\right] > 1 - \delta.
    \end{equation*}
    Since the bounds $\{\mu_n<\text{UCB}_{t,a,n}\}_{a\in\mathcal{A}}$ hold simultaneously, this set includes the bound $\mu_n<\min_{a\in\mathcal{A}}\text{UCB}_{t,a,n}$.
    \begin{equation*}
        \forall \delta \in (0,1), \;\mathbb{P}\left[\forall n\in\mathcal{N},t\in\mathcal{T},\mu_n<\min_{a\in\mathcal{A}}\text{UCB}_{t,a,n}\right] > 1 - \delta
    \end{equation*}
    Using the definition of $\text{UCB}_{t,n}$ given in Equation~\ref{eq:min-UCB_UCB_t,n} gives
    \begin{equation*}
        \forall\delta\in(0,1),\;\mathbb{P}\left[\forall n\in\mathcal{N},t\in\mathcal{T},\mu_n<\text{UCB}_{t,n}\right]>1-\delta.
    \end{equation*}
\end{proof}

\subsection{\textsc{Min-Width} Weights Derivation}
\label{appendix:min-width-weights}
\begin{proposition}[\textsc{Min-Width} Weights Derivation]
    Suppose agent $a$ pulls arm $n$ a fixed number of times, a number we denote $c_{a,n}$, where the reward from each pull is $Y_{i,a,n}\sim\text{Bern}(s_a\mu_n)$. Let $\mathcal{C}_n=\{c_{a,n}\}_{a=1}^A$ contain the $c_{a,n}$ for every agent. Let $D_{\mathcal{C}_n,n}$ be the weighted sum of the independent rewards collected by all the agents from arm $n$, expressed in terms of weights $w_{\mathcal{C}_n,a,n}$.
\begin{equation*}
    D_{\mathcal{C}_n,n}=\sum_{a=1}^Aw_{\mathcal{C}_n,a,n}\sum_{i=1}^{c_{a,n}}Y_{i,a,n}=\sum_{a=1}^A\sum_{i=1}^{c_{a,n}}w_{\mathcal{C}_n,a,n}Y_{i,a,n}
\end{equation*}
Then the weights $w_{\mathcal{C}_n,a,n}$ that minimize the width of the confidence interval on $\mu_n$ given by
\begin{equation*}
    \gamma_{\mathcal{C}_n,n}=\sqrt{\frac{\ln(2/\delta)\sum_{a=1}^A{w_{\mathcal{C}_n,a,n}}^2c_{a,n}}{2}}
\end{equation*}
under the constraint that the empirical estimator $D_{\mathcal{C}_n,n}$ is unbiased are given by
    \begin{equation*}
        w_{\mathcal{C}_n,a,n}=\frac{s_a}{\sum_{b=1}^A{s_b}^2c_{b,n}}\mathbbm{1}_{c_{a,n}>0}.
    \end{equation*}
\end{proposition}

\begin{proof}
     Since the rewards collected by a certain agent when pulling a certain arm are independent and identically distributed, the order of the rewards in such a sequence is unimportant. Thus, we consider weights on such sequences of i.i.d. rewards rather than on every reward itself. Note that if $c_{a,n}=0$, then agent $a$ has collected no rewards for arm $n$, and so we set $w_{\mathcal{C}_n,a,n}=0$ for any such agent. Hence, we need to solve for $w_{\mathcal{C}_n,a,n}$ only for agents with $c_{a,n}>0$. Consider the expectation of $D_{\mathcal{C}_n,n}$ as defined in Equation~\ref{eq:D_Cn,n}.
\begin{equation*}
    \mathbb{E}\left[D_{\mathcal{C}_n,n}\right]=\sum_{a=1}^A\sum_{i=1}^{c_{a,n}}w_{\mathcal{C}_n,a,n}s_a\mu_n=\mu_n\sum_{a=1}^Aw_{\mathcal{C}_n,a,n}s_ac_{a,n}
\end{equation*}
If $D_{\mathcal{C}_n,n}$ is to be unbiased, then we need
\begin{equation*}
    \mathbb{E}\left[D_{\mathcal{C}_n,n}\right]=\mu_n,
\end{equation*}
which sets the constraint
\begin{equation}
    \sum_{a=1}^Aw_{\mathcal{C}_n,a,n}s_ac_{a,n}=1.
    \label{eq:unbiased_constraint}
\end{equation}
The random variable being summed in Equation \ref{eq:D_Cn,n} is $w_{\mathcal{C}_n,a,n}Y_{i,a,n}$, which is bounded by
\begin{equation*}
    0\leq w_{\mathcal{C}_n,a,n}Y_{i,a,n}\leq w_{\mathcal{C}_n,a,n}.
\end{equation*}
Hoeffding's inequality then gives
\begin{equation*}
    \forall\alpha>0,\;\mathbb{P}\big[\big|D_{\mathcal{C}_n,n}-\mathbb{E}\big[D_{\mathcal{C}_n,n}\big]\big|\geq\alpha\big]\\\leq2\exp\left(-\frac{2\alpha^2}{\sum_{a=1}^A\sum_{i=1}^{c_{a,n}}{w_{\mathcal{C}_n,a,n}}^2}\right).
\end{equation*}
Collapsing the sum over $i$ and using the unbiasedness of $D_{\mathcal{C}_n,n}$ yields
\begin{equation*}
    \forall\alpha>0,\;\mathbb{P}\big[\big|D_{\mathcal{C}_n,n}-\mu_n\big|\geq\alpha\big]\\\leq2\exp\left(-\frac{2\alpha^2}{\sum_{a=1}^A{w_{\mathcal{C}_n,a,n}}^2c_{a,n}}\right).
\end{equation*}
Taking the complement of the equation results in
\begin{equation*}
    \forall\alpha>0,\;\mathbb{P}\big[\big|D_{\mathcal{C}_n,n}-\mu_n\big|<\alpha\big]\\>1-2\exp\left(-\frac{2\alpha^2}{\sum_{a=1}^A{w_{\mathcal{C}_n,a,n}}^2c_{a,n}}\right).
\end{equation*}
Set
\begin{equation*}
    \delta=2\exp\left(-\frac{2\alpha^2}{\sum_{a=1}^A{w_{\mathcal{C}_n,a,n}}^2c_{a,n}}\right).
\end{equation*}
We can now solve for $\alpha$ in terms of $\delta$.
\begin{multline*}
    \exp\left(-\frac{2\alpha^2}{\sum_{a=1}^A{w_{\mathcal{C}_n,a,n}}^2c_{a,n}}\right)=\frac{\delta}{2}\Longrightarrow-\frac{2\alpha^2}{\sum_{a=1}^A{w_{\mathcal{C}_n,a,n}}^2c_{a,n}}=\ln\left(\frac{\delta}{2}\right)\Longrightarrow\frac{2\alpha^2}{\sum_{a=1}^A{w_{\mathcal{C}_n,a,n}}^2c_{a,n}}=\ln\left(\frac{2}{\delta}\right)\\\Longrightarrow\alpha=\sqrt{\frac{\ln(2/\delta)\sum_{a=1}^A{w_{\mathcal{C}_n,a,n}}^2c_{a,n}}{2}}
\end{multline*}
The Hoeffding inequality can then be expressed in terms of $\delta$ rather than $\alpha$.
\begin{equation*}
    \forall\delta\in(0,1),\;\mathbb{P}\Biggr[\big|D_{\mathcal{C}_n,n}-\mu_n\big|\\<\sqrt{\frac{\ln(2/\delta)\sum_{a=1}^A{w_{\mathcal{C}_n,a,n}}^2c_{a,n}}{2}}\Biggr]>1-\delta
\end{equation*}
Let $\gamma_{\mathcal{C}_n,n}$ be the width of the confidence interval on the mean of arm $n$ for some fixed number of pulls of each arm by each agent captured in $\mathcal{C}_n$.
\begin{equation*}
    \gamma_{\mathcal{C}_n,n}=\sqrt{\frac{\ln(2/\delta)\sum_{a=1}^A{w_{\mathcal{C}_n,a,n}}^2c_{a,n}}{2}}
\end{equation*}
Let
\begin{equation*}
    \beta = \sqrt{\frac{\ln(2/\delta)}{2}}.
\end{equation*}
The width $\gamma_{\mathcal{C}_n,n}$ can then be expressed as
\begin{equation*}
    \gamma_{\mathcal{C}_n,n}=\beta\sqrt{\sum_{a=1}^A{w_{\mathcal{C}_n,a,n}}^2c_{a,n}}.
\end{equation*}
We want the confidence interval to be as tight as possible around the unbiased estimator $D_{\mathcal{C}_n,n}$, so we solve for the weights that minimize $\gamma_{\mathcal{C}_n,n}$ subject to the constraint that $D_{\mathcal{C}_n,n}$ is unbiased for any non-random $\mathcal{C}_n$. We do this with the method of Lagrange multipliers. Define the Lagrangian
\begin{equation*}
    \mathcal{L}(w,\lambda)=\gamma_{\mathcal{C}_n,n}+\lambda\left(\sum_{b=1}^Aw_{\mathcal{C}_n,b,n}s_bc_{b,n}-1\right)\\=\beta\sqrt{\sum_{b=1}^A{w_{\mathcal{C}_n,b,n}}^2c_{b,n}}+\lambda\left(\sum_{b=1}^Aw_{\mathcal{C}_n,b,n}s_bc_{b,n}-1\right).
\end{equation*}
Consider the derivative of $\mathcal{L}$ with respect to $w_{\mathcal{C}_n,a,n}$ for some agent $a$ with $c_{a,n}>0$.
\begin{equation*}
    \frac{\partial\mathcal{L}}{\partial w_{\mathcal{C}_n,a,n}}=\frac{2\beta w_{\mathcal{C}_n,a,n}c_{a,n}}{2\sqrt{\sum_{b=1}^A{w_{\mathcal{C}_n,b,n}}^2c_{b,n}}}+\lambda s_ac_{a,n}\\=\frac{\beta w_{\mathcal{C}_n,a,n}c_{a,n}}{\sqrt{\sum_{b=1}^A{w_{\mathcal{C}_n,b,n}}^2c_{b,n}}}+\lambda s_ac_{a,n}
\end{equation*}
Next, consider the derivative of $\mathcal{L}$ with respect to $\lambda$.
\begin{equation*}
    \frac{\partial\mathcal{L}}{\partial\lambda}=\sum_{b=1}^Aw_{\mathcal{C}_n,b,n}s_bc_{b,n}-1
\end{equation*}
To solve the optimization problem, we set both derivatives equal to 0. First, we set the derivative with respect to the weights equal to 0.
\begin{multline}
    \frac{\beta w_{\mathcal{C}_n,a,n}c_{a,n}}{\sqrt{\sum_{b=1}^A{w_{\mathcal{C}_n,b,n}}^2c_{b,n}}}+\lambda s_ac_{a,n}=0\Longrightarrow \frac{\beta w_{\mathcal{C}_n,a,n}}{\sqrt{\sum_{b=1}^A{w_{\mathcal{C}_n,b,n}}^2c_{b,n}}}+\lambda s_a=0\Longrightarrow\frac{\beta w_{\mathcal{C}_n,a,n}}{\sqrt{\sum_{b=1}^A{w_{\mathcal{C}_n,b,n}}^2c_{b,n}}}=-\lambda s_a
    \\\Longrightarrow w_{\mathcal{C}_n,a,n}=-\frac{\lambda s_a\sqrt{\sum_{b=1}^A{w_{\mathcal{C}_n,b,n}}^2c_{b,n}}}{\beta}
    \label{eq:weights_with_lambda}
\end{multline}
Then, we set the derivative with respect to $\lambda$ equal to 0.
\begin{equation*}
    \sum_{b=1}^Aw_{\mathcal{C}_n,b,n}s_bc_{b,n}-1=0\Longrightarrow\sum_{b=1}^Aw_{\mathcal{C}_n,b,n}s_bc_{b,n}=1
    \label{eq:lambda_derivative}
\end{equation*}
We can plug the weights expressed in terms of $\lambda$, given in Equation \ref{eq:weights_with_lambda}, into Equation \ref{eq:lambda_derivative} in order to solve for $\lambda$.
\begin{multline*}
    \sum_{b=1}^A\left(-\frac{\lambda s_b\sqrt{\sum_{d=1}^A{w_{\mathcal{C}_n,d,n}}^2c_{d,n}}}{\beta}\right)s_bc_{b,n}=1\Longrightarrow-\frac{\lambda\sqrt{\sum_{d=1}^A{w_{\mathcal{C}_n,d,n}}^2c_{d,n}}}{\beta}\sum_{b=1}^A{s_b}^2c_{b,n}=1
    \\\Longrightarrow\lambda=-\frac{\beta}{\sqrt{\sum_{d=1}^A{w_{\mathcal{C}_n,d,n}}^2c_{d,n}}\sum_{b=1}^A{s_b}^2c_{b,n}}
\end{multline*}
We now plug this expression for $\lambda$ back into the expression for the weights given in Equation \ref{eq:weights_with_lambda}, yielding
\begin{equation*}
     w_{\mathcal{C}_n,a,n}=-\frac{\beta}{\sqrt{\sum_{d=1}^A{w_{\mathcal{C}_n,d,n}}^2c_{d,n}}\sum_{b=1}^A{s_b}^2c_{b,n}}\\\times-\frac{ s_a\sqrt{\sum_{b=1}^A{w_{\mathcal{C}_n,b,n}}^2c_{b,n}}}{\beta}=\frac{s_a}{\sum_{b=1}^A{s_b}^2c_{b,n}},
\end{equation*}
which holds for any agent $a$ with $c_{a,n}>0$. Since $w_{\mathcal{C}_n,a,n}=0$ for agents with $c_{a,n}=0$, the final form for the weights can be captured as
\begin{equation*}
    w_{\mathcal{C}_n,a,n}=\frac{s_a}{\sum_{b=1}^A{s_b}^2c_{b,n}}\mathbbm{1}_{c_{a,n}>0}.
\end{equation*}
\end{proof}

\subsection{\textsc{Min-Width} Concentration Bound}
\label{appendix:min-width-concentration-bound}
\begin{theorem}[\textsc{Min-Width} Concentration Bound]
    Suppose the empirical estimator of the mean of arm $n$ at time $t\in\mathcal{T}$ is given by
    \begin{equation*}
        \hat{\mu}_{t,n} =
        \begin{cases}
            t\notin\mathcal{T}_n & 0.5 \\
            t\in\mathcal{T}_n & \frac{1}{\sum_{b=1}^A{s_b}^2c_{t,b,n}}\sum_{a=1}^A s_a \sum_{\tau=1}^t\mathbbm{1}_{f_\tau(a)=n} Y_{\tau,a,n}
        \end{cases}
    \end{equation*}
    for $Y_{\tau,a,n}$, $c_{t,b,n}$, and $\mathcal{T}_n$ defined in Equations~\ref{eq:Y_tan}, \ref{eq:c_tan}, and \ref{eq:T_n}, respectively. Then $\hat{\mu}_{t,n}$ satisfies
    \begin{equation*}
        \forall\delta\in(0,1),\;\mathbb{P}\left[\forall n\in\mathcal{N},t\in\mathcal{T},|\hat{\mu}_{t,n}-\mu_n|<\epsilon_{t,n}\right]>1-\delta
    \end{equation*}
    for
    \begin{equation*}
        \epsilon_{t,n} = \sqrt{\frac{\ln(2NG(T,A)/\delta)}{2\sum_{a=1}^A {s_a}^2 c_{t,a,n}}},
    \end{equation*}
    and
    \begin{equation*}
        G(T,A)=\sum_{t=1}^T\binom{t+A-1}{A-1}
    \end{equation*}
    where
    \begin{equation*}
        G(T,A)<(T+1)^A.
    \end{equation*}
\end{theorem}

\begin{proof}
Let us begin from Equation~\ref{eq:min-width-Hoeffding-no-union}, which says
\begin{equation*}
    \forall\delta\in(0,1),\;\mathbb{P}\Biggr[\big|D_{\mathcal{C}_n,n}-\mu_n\big|\\<\sqrt{\frac{\ln(2/\delta)\sum_{a=1}^A{w_{\mathcal{C}_n,a,n}}^2c_{a,n}}{2}}\Biggr]>1-\delta.
\end{equation*}
Applying a union bound over
\\$\biggl\{\big|D_{\mathcal{C}_n,n}-\mu_n\big|<\sqrt{\ln(2/\delta)\sum_{a=1}^A{w_{\mathcal{C}_n,a,n}}^2c_{a,n}/2}\biggl\}_{n\in\mathcal{N}}$ yields
\begin{equation*}
    \forall\delta\in(0,1),\;\mathbb{P}\Biggr[\forall n\in\mathcal{N},\big|D_{\mathcal{C}_n,n}-\mu_n\big|\\<\sqrt{\frac{\ln(2N/\delta)\sum_{a=1}^A{w_{\mathcal{C}_n,a,n}}^2c_{a,n}}{2}}\Biggr]>1-\delta.
\end{equation*}
Let $\mathcal{H}$ be the set of all possible instantiations of the set $\mathcal{C}_n$  assuming that arm $n$ has been pulled at least once within a time horizon of $T$, so $\mathcal{H}$ is a set of sets. We need to apply a union bound over
\\$\biggl\{\big|D_{\mathcal{C}_n,n}-\mu_n\big|<\sqrt{\frac{\ln(2N/\delta)\sum_{a=1}^A{w_{\mathcal{C}_n,a,n}}^2c_{a,n}}{2}}\biggl\}_{\mathcal{C}_n\in\mathcal{H}}$. In order to perform this union bound, we must determine the cardinality of $\mathcal{H}$, which is the number of possible instantiations of the set $\mathcal{C}_n$. $\mathcal{C}_n$ contains $A$ elements, where each element $c_{a,n}$ is the number of times agent $a$ has pulled arm $n$ through time $T$. If the agent has never pulled arm $n$, then that element is 0, and if the agent pulls arm $n$ in every time step, then the element is $T$. This means that each element of the set is an integer between 0 and $T$. By the assumption that arm $n$ has been pulled at least once, the sum of the elements in $\mathcal{C}_n$ must be at least 1. Moreover, because a maximum of one agent can pull each arm in every time step, the sum of the elements can be no greater than $T$. Hence, the sum of the elements in $\mathcal{C}_n$ must be an integer between 1 and $T$. In order to break down this complex counting problem, let us first consider the number of permissible instantiations of the set that have the elements sum to some $1\leq t\leq T$. Using the terminology of the stars and bars method, this problem is equivalent to the number of ways to distribute $t$ stars among $A$ bars, where empty bins are allowed. The number of possible arrangements is a well-known result: $\binom{t+A-1}{A-1}$. Since $t$ can range from 1 to $T$, the number of possible instantiations of the set $\mathcal{C}_n$ is given by summing this result from 1 to $T$, a quantity we denote $G(T,A)$, given in Equation~\ref{eq:G(T,A)}. If we ignore the constraint on the sum of $\mathcal{C}_n$, then the number of possible instantiations would simply be $(T+1)^A$ since $\mathcal{C}_n$ has $A$ elements, each of which can take on any integer value between 0 and $T$. This means that the cardinality of $\mathcal{H}$, which we have denoted $G(T,A)$, is bounded above by $(T+1)^A$. Note that if all the agents were identical, then $G(T,A)=T$ as in \textsc{CUCB}. The outcome of performing the union bound is then
\begin{equation*}
    \forall\delta\in(0,1),\;\mathbb{P}\Biggr[\forall n\in\mathcal{N},\mathcal{C}_n\in\mathcal{H},\big|D_{\mathcal{C}_n,n}-\mu_n\big|\\<\sqrt{\frac{\ln(2NG(T,A)/\delta)\sum_{a=1}^A{w_{\mathcal{C}_n,a,n}}^2c_{a,n}}{2}}\Biggr]>1-\delta.
\end{equation*}
After plugging in the derived weights from Equation~\ref{eq:weights}, we get
\begin{equation*}
    \forall\delta\in(0,1),\;\mathbb{P}\Biggr[\forall n\in\mathcal{N},\mathcal{C}_n\in\mathcal{H},\big|D_{\mathcal{C}_n,n}-\mu_n\big|\\<\sqrt{\frac{\ln(2NG(T,A)/\delta)\sum_{a=1}^A\left(\frac{s_a}{\sum_{b=1}^A{s_b}^2c_{b,n}}\mathbbm{1}_{c_{a,n}>0}\right)^2c_{a,n}}{2}}\Biggr]\\>1-\delta.
\end{equation*}
Expanding the square gives
\begin{equation*}
    \forall\delta\in(0,1),\;\mathbb{P}\Biggr[\forall n\in\mathcal{N},\mathcal{C}_n\in\mathcal{H},\big|D_{\mathcal{C}_n,n}-\mu_n\big|\\<\frac{1}{\sum_{b=1}^A{s_b}^2c_{b,n}}\sqrt{\frac{\ln(2NG(T,A)/\delta)\sum_{a=1}^A\mathbbm{1}_{c_{a,n}>0}{s_a}^2c_{a,n}}{2}}\Biggr]\\>1-\delta.
\end{equation*}
Because $\mathbbm{1}_{c_{a,n}>0}$ is being multiplied by $c_{a,n}$, the indicator variable is unnecessary, and we get
\begin{equation*}
    \forall\delta\in(0,1),\;\mathbb{P}\Biggr[\forall n\in\mathcal{N},\mathcal{C}_n\in\mathcal{H},\big|D_{\mathcal{C}_n,n}-\mu_n\big|\\<\sqrt{\frac{\ln(2NG(T,A)/\delta)}{2\sum_{a=1}^A{s_a}^2c_{a,n}}}\Biggr]>1-\delta.
\end{equation*}
Since $\mathcal{C}_{t,n}\in\mathcal{H}\;\forall t\in\mathcal{T}_n$,
\begin{equation*}
    \forall\delta\in(0,1),\;\mathbb{P}\Biggr[\forall n\in\mathcal{N},t\in\mathcal{T}_n,\big|D_{\mathcal{C}_{t,n},n}-\mu_n\big|\\<\sqrt{\frac{\ln(2NG(T,A)/\delta)}{2\sum_{a=1}^A{s_a}^2c_{t,a,n}}}\Biggr]>1-\delta.
\end{equation*}
Now consider what $D_{\mathcal{C}_{t,n},n}$ simplifies to, recalling the definition of $D_{\mathcal{C}_n,n}$ in Equation~\ref{eq:D_Cn,n} and using $w_{\mathcal{C}_{t,n},a,n}$ from Equation~\ref{eq:weights}.
\begin{equation*}
    D_{\mathcal{C}_{t,n},n}=\sum_{a=1}^Aw_{\mathcal{C}_{t,n},a,n}\sum_{i=1}^{c_{t,a,n}}Y_{i,a,n}\\=\sum_{a=1}^A\frac{s_a}{\sum_{b=1}^A{s_b}^2c_{t,b,n}}\mathbbm{1}_{c_{t,a,n}>0}\sum_{i=1}^{c_{t,a,n}}Y_{i,a,n}\\=\frac{1}{\sum_{b=1}^A{s_b}^2c_{t,b,n}}\sum_{a=1}^A\mathbbm{1}_{c_{t,a,n}>0}\;s_a\sum_{i=1}^{c_{t,a,n}}Y_{i,a,n}
\end{equation*}
The Hoeffding inequality then becomes
\begin{equation*}
    \forall\delta\in(0,1),\;\mathbb{P}\Biggr[\forall n\in\mathcal{N},t\in\mathcal{T}_n,\\\Biggr|\frac{1}{\sum_{b=1}^A{s_b}^2c_{t,b,n}}\sum_{a=1}^A\mathbbm{1}_{c_{t,a,n}>0}\;s_a\sum_{i=1}^{c_{t,a,n}}Y_{i,a,n}-\mu_n\Biggr|\\<\sqrt{\frac{\ln(2NG(T,A)/\delta)}{2\sum_{a=1}^A{s_a}^2c_{t,a,n}}}\Biggr]>1-\delta.
\end{equation*}
By the definition of $c_{t,a,n}$ in Equation~\ref{eq:c_tan}, summing over $i$ from 1 to $c_{t,a,n}$ only for agents with $c_{t,a,n}>0$ is equivalent to summing over $\tau$ from 1 to $t$ and multiplying by the indicator variable $\mathbbm{1}_{f_\tau(a)=n}$.
\begin{equation*}
    \forall\delta\in(0,1),\;\mathbb{P}\Biggr[\forall n\in\mathcal{N},t\in\mathcal{T}_n,\\\Biggr|\frac{1}{\sum_{b=1}^A{s_b}^2c_{t,b,n}}\sum_{a=1}^As_a\sum_{\tau=1}^t \mathbbm{1}_{f_\tau(a)=n} Y_{\tau,a,n}-\mu_n\Biggr|\\<\sqrt{\frac{\ln(2NG(T,A)/\delta)}{2\sum_{a=1}^A{s_a}^2c_{t,a,n}}}\Biggr]>1-\delta.
\end{equation*}
Using the definitions of $\hat{\mu}_{t,n}$ and $\epsilon_{t,n}$ given in Equations~\ref{eq:muhat_t,n} and \ref{eq:epsilon_t,n}, respectively, gives
\begin{equation*}
    \forall\delta\in(0,1),\;\mathbb{P}\left[\forall n\in\mathcal{N},t\in\mathcal{T}_n,|\hat{\mu}_{t,n}-\mu_n|<\epsilon_{t,n}\right]>1-\delta.
\end{equation*}
Recall that for $t\notin\mathcal{T}_n$, we have $\hat{\mu}_{t,n}=0.5$ by Equation~\ref{eq:muhat_t,n} and $\epsilon_{t,n}=\infty$ since $c_{t,n}=0$. Because the difference between the true mean $\mu_n$ and 0.5 is less than infinity with probability 1, we get
\begin{equation*}
    \forall\delta\in(0,1),\;\mathbb{P}\left[\forall n\in\mathcal{N},t\in\mathcal{T},|\hat{\mu}_{t,n}-\mu_n|<\epsilon_{t,n}\right]>1-\delta.
\end{equation*}
\end{proof}

\subsection{Lemma}
\label{appendix:lemma}
\begin{lemma}
    Suppose $N$, $T$, $A$, and $c_{t,f_t(a)}$ are positive integers as defined in Section 3. Then the following bound holds.
    \begin{equation*}
        \sum_{t=N}^T\sum_{a=1}^A\frac{1}{\sqrt{c_{t,f_t(a)}}}< 2\sqrt{ANT}
    \end{equation*}
\end{lemma}

\begin{proof}
\begin{equation}
    \sum_{t=N}^T\sum_{a=1}^A\frac{1}{\sqrt{c_{t,f_t(a)}}}\leq\sum_{t=1}^T\sum_{a=1}^A\frac{1}{\sqrt{c_{t,f_t(a)}}}\\=\sum_{n=1}^N\sum_{t=1}^T\sum_{a=1}^A\mathbbm{1}_{f_t(a)=n}\frac{1}{\sqrt{c_{t,n}}}=\sum_{n=1}^N\sum_{j=1}^{c_{T,n}}\frac{1}{\sqrt{j}}
    \label{eq:sum_to_c_Tn}
\end{equation}
Because $\frac{1}{\sqrt{j}}$ is monotonically decreasing on $[1,c_{T,n}]$, the right Riemann sum $\sum_{j=1}^{c_{T,n}}\frac{1}{\sqrt{j}}$ is an underestimation of $1+\int_1^{c_{T,n}}\frac{1}{\sqrt{x}}dx$.
    \begin{equation}
        \sum_{j=1}^{c_{T,n}}\frac{1}{\sqrt{j}}\leq 1+\int_1^{c_{T,n}}\frac{1}{\sqrt{x}}dx=1+2\sqrt{x}|_1^{c_{T,n}}\\=1+2\sqrt{c_{T,n}}-2=2\sqrt{c_{T,n}}-1<2\sqrt{c_{T,n}}
        \label{eq:bound_with_sqrt_c_Tn}
    \end{equation}
Applying the bound in Equation \ref{eq:bound_with_sqrt_c_Tn} to Equation \ref{eq:sum_to_c_Tn} gives
\begin{equation}
    \sum_{t=N}^T\sum_{a=1}^A\frac{1}{\sqrt{c_{t,f_t(a)}}}<2\sum_{n=1}^N\sqrt{c_{T,n}}.
    \label{eq:sum_over_n_sqrt_c_Tn}
\end{equation}
We can use the form of Jensen's inequality for a concave function to bound $\sum_{n=1}^N\sqrt{c_{T,n}}$.
\begin{equation*}
    \frac{\sum_{n=1}^N\sqrt{c_{T,n}}}{\sum_{n'=1}^N1}\leq\sqrt{\frac{\sum_{n=1}^Nc_{T,n}}{\sum_{n'=1}^N1}}
\end{equation*}
$\sum_{n=1}^Nc_{T,n}$ is the total number of pulls performed through time $T$, which is $AT$ since each agent pulls one arm in every time step.
\begin{equation}
    \frac{\sum_{n=1}^N\sqrt{c_{T,n}}}{N}\leq\sqrt{\frac{AT}{N}}\Longrightarrow\sum_{n=1}^N\sqrt{c_{T,n}}\leq\sqrt{ANT}
    \label{eq:ANT_bound}
\end{equation}
Applying the bound in Equation \ref{eq:ANT_bound} to Equation \ref{eq:sum_over_n_sqrt_c_Tn} completes the proof.
\end{proof}

\subsection{\textsc{Min-Width} Regret Bound}
\label{appendix:min-width-regret}
\begin{theorem}[\textsc{Min-Width} Regret Bound]
    Suppose we act according to the \textsc{Min-Width} algorithm. Then the cumulative regret at time $T$ is bounded with high probability according to
    \begin{equation*}
        \forall\delta\in(0,1),\mathbb{P}\Biggr[R_T<A(N-1)\\+2\sqrt{2ANT\ln(2NG(T,A)/\delta)}\frac{\max\mathcal{S}}{\min\mathcal{S}}\Biggr]>1-\delta.
    \end{equation*}
\end{theorem}

\begin{proof}
    First, we can bound $\mu_n-\hat{\mu}_{t,n}$ by
    \begin{equation*}
        \mu_n-\hat{\mu}_{t,n}\leq|\mu_n-\hat{\mu}_{t,n}|=|\hat{\mu}_{t,n}-\mu_n|.
    \end{equation*}
    Using the bound on $|\hat{\mu}_{t,n}-\mu_n|$ from Equation~\ref{eq:mu^hat_tn_bound} gives
    \begin{equation*}
        \forall\delta\in(0,1),\mathbb{P}\left[\forall n\in\mathcal{N},t\in\mathcal{T}_n,\mu_n-\hat{\mu}_{t,n}<\epsilon_{t,n}\right]>1-\delta.
    \end{equation*}
    This is equivalent to
    \begin{equation*}
        \forall\delta\in(0,1),\mathbb{P}\left[\forall n\in\mathcal{N},t\in\mathcal{T},\mu_n<\hat{\mu}_{t,n}+\epsilon_{t,n}\right]>1-\delta.
    \end{equation*}
    Recall the UCB on the mean of arm $n$ at time $t\in\mathcal{T}$ given in Equation~\ref{eq:min-width-UCB_t,n} as
    \begin{equation*}
        \text{UCB}_{t,n}=\hat{\mu}_{t,n}+\epsilon_{t,n}
    \end{equation*}
    with $\hat{\mu}_{t,n}$ and $\epsilon_{t,n}$ given in Equations~\ref{eq:muhat_t,n} and \ref{eq:epsilon_t,n}, respectively. We can then bound $\mu_n$ by this UCB with high probability.
    \begin{equation*}
        \forall\delta\in(0,1),\mathbb{P}\left[\forall n\in\mathcal{N},t\in\mathcal{T},\mu_n<\text{UCB}_{t,n}\right]>1-\delta.
    \end{equation*}
    Recall the definition of cumulative regret from Equation~\ref{eq:regret_def}.
    \begin{equation*}
        R_T=\sum_{t=1}^T\sum_{a=1}^As_a\left(\mu_{f^\star(a)}-\mu_{f_t(a)}\right)
    \end{equation*}
    Let us split the regret into terms with $t<N$ and terms with $t\geq N$.
    \begin{equation*}
        R_T=\sum_{t=1}^{N-1}\sum_{a=1}^As_a\left(\mu_{f^\star(a)}-\mu_{f_t(a)}\right)\\+\sum_{t=N}^T\sum_{a=1}^As_a\left(\mu_{f^\star(a)}-\mu_{f_t(a)}\right)
    \end{equation*}
    Because $0<\mu_n<1$ $\forall n\in\mathcal{N}$, the difference between the means of any two arms is bounded by
    \begin{equation*}
        \mu_n-\mu_{n'}<1\;\forall n,n'\in\mathcal{N}.
    \end{equation*}
    Applying this bound to the first term in Equation~\ref{eq:R_T_split} with $n=f^\star(a)$ and $n'=f_t(a)$ gives
    \begin{equation*}
        R_T<\sum_{t=1}^{N-1}\sum_{a=1}^As_a+\sum_{t=N}^T\sum_{a=1}^As_a\left(\mu_{f^\star(a)}-\mu_{f_t(a)}\right).
    \end{equation*}
    Next, note that $s_a\leq1$ $\forall a\in\mathcal{A}$.
    \begin{equation*}
        R_T<\sum_{t=1}^{N-1}\sum_{a=1}^A1+\sum_{t=N}^T\sum_{a=1}^As_a\left(\mu_{f^\star(a)}-\mu_{f_t(a)}\right)
    \end{equation*}
    The first term then reduces so that
    \begin{equation*}
        R_T<A(N-1)+\sum_{t=N}^T\sum_{a=1}^As_a\left(\mu_{f^\star(a)}-\mu_{f_t(a)}\right).
    \end{equation*}
    Now, we need to bound the second term. For ease of notation, define $R_{N:T}$ as
    \begin{equation*}
        R_{N:T}=\sum_{t=N}^T\sum_{a=1}^As_a\left(\mu_{f^\star(a)}-\mu_{f_t(a)}\right)\\=\sum_{t=N}^T\left(\sum_{a=1}^As_a\mu_{f^\star(a)}-\sum_{a=1}^As_a\mu_{f_t(a)}\right).
    \end{equation*}
    Using Equation \ref{eq:mu_n_less_UCB} with $n=f^\star(a)$ to bound $\mu_{f^\star(a)}$ gives
    \begin{equation*}
        \forall\delta\in(0,1),\mathbb{P}\Biggr[R_{N:T}<\sum_{t=N}^T\Biggr(\sum_{a=1}^As_a\text{UCB}_{t,f^\star(a)}\\-\sum_{a=1}^As_a\mu_{f_t(a)}\Biggr)\Biggr]>1-\delta.
    \end{equation*}
    By construction, for all $t$ the \textsc{Min-Width} algorithm selects a configuration $f$ that maximizes $\sum_{a=1}^As_a$UCB$_{t,f_t(a)}$. This implies
    \begin{equation*}
        \forall t\in\mathcal{T},\sum_{a=1}^As_a\text{UCB}_{t,f^\star(a)}\leq\sum_{a=1}^As_a\text{UCB}_{t,f_t(a)}.
    \end{equation*}
    Using Equation \ref{eq:using_alg} in Equation \ref{eq:regret_3} gives
    \begin{equation*}
        \forall\delta\in(0,1),\mathbb{P}\Biggr[R_{N:T}<\sum_{t=N}^T\Biggr(\sum_{a=1}^As_a\text{UCB}_{t,f_t(a)}\\-\sum_{a=1}^As_a\mu_{f_t(a)}\Biggr)\Biggr]>1-\delta.
    \end{equation*}
    Plugging in Equation \ref{eq:min-width-UCB_t,n} with $n=f_t(a)$ gives
    \begin{equation*}
        \forall\delta\in(0,1),\mathbb{P}\Biggr[R_{N:T}<\sum_{t=N}^T\Biggr(\sum_{a=1}^As_a\left(\hat{\mu}_{t,f_t(a)}+\epsilon_{t,f_t(a)}\right)\\-\sum_{a=1}^As_a\mu_{f_t(a)}\Biggr)\Biggr]>1-\delta.
    \end{equation*}
    Regrouping terms yields
    \begin{equation*}
        \forall\delta\in(0,1),\mathbb{P}\Biggr[R_{N:T}<\sum_{t=N}^T\sum_{a=1}^As_a\left(\hat{\mu}_{t,f_t(a)}+\epsilon_{t,f_t(a)}-\mu_{f_t(a)}\right)\Biggr]\\>1-\delta.
    \end{equation*}
    Since $\hat{\mu}_{t,f_t(a)}-\mu_{f_t(a)}\leq|\hat{\mu}_{t,f_t(a)}-\mu_{f_t(a)}|$, we get
    \begin{equation*}
        \forall\delta\in(0,1),\mathbb{P}\Biggr[R_{N:T}<\sum_{t=N}^T\sum_{a=1}^As_a\left(\big|\hat{\mu}_{t,f_t(a)}-\mu_{f_t(a)}\big|+\epsilon_{t,f_t(a)}\right)\Biggr]\\>1-\delta.
    \end{equation*}
    Using Equation~\ref{eq:mu^hat_tn_bound} for $n=f_t(a)$ gives
    \begin{equation*}
        \forall\delta\in(0,1),\mathbb{P}\Biggr[R_{N:T}<\sum_{t=N}^T\sum_{a=1}^As_a\left(\epsilon_{t,f_t(a)}+\epsilon_{t,f_t(a)}\right)\Biggr]>1-\delta.
    \end{equation*}
    This reduces to
    \begin{equation*}
        \forall\delta\in(0,1),\mathbb{P}\Biggr[R_{N:T}<2\sum_{t=N}^T\sum_{a=1}^As_a\epsilon_{t,f_t(a)}\Biggr]>1-\delta.
    \end{equation*}
    Now, we plug in for $\epsilon_{t,f_t(a)}$ using Equation~\ref{eq:epsilon_t,n} with $n=f_t(a)$.
    \begin{equation*}
        \forall\delta\in(0,1),\mathbb{P}\Biggr[R_{N:T}<2\sum_{t=N}^T\sum_{a=1}^As_a\sqrt{\frac{\ln(2NG(T,A)/\delta)}{2\sum_{b=1}^A {s_b}^2 c_{t,b,f_t(a)}}}\Biggr]>1-\delta.
    \end{equation*}
    Pulling out constants gives
    \begin{equation*}
        \forall\delta\in(0,1),\mathbb{P}\Biggr[R_{N:T}<\sqrt{2\ln(2NG(T,A)/\delta)}\\\sum_{t=N}^T\sum_{a=1}^A\frac{s_a}{\sqrt{\sum_{b=1}^A {s_b}^2 c_{t,b,f_t(a)}}}\Biggr]>1-\delta.
    \end{equation*}
    Note that $\forall a\in\mathcal{A},\;s_a\leq\max\mathcal{S}$ and $\forall b\in\mathcal{A},\;s_b\geq\min\mathcal{S}$.
    \begin{equation*}
        \forall\delta\in(0,1),\mathbb{P}\Biggr[R_{N:T}<\sqrt{2\ln(2NG(T,A)/\delta)}\frac{\max\mathcal{S}}{\min\mathcal{S}}\\\sum_{t=N}^T\sum_{a=1}^A\frac{1}{\sqrt{\sum_{b=1}^A c_{t,b,f_t(a)}}}\Biggr]>1-\delta.
    \end{equation*}
    Using Equation~\ref{eq:c_tn} for $n=f_t(a)$ yields
    \begin{equation*}
        \forall\delta\in(0,1),\mathbb{P}\Biggr[R_{N:T}<\sqrt{2\ln(2NG(T,A)/\delta)}\frac{\max\mathcal{S}}{\min\mathcal{S}}\\\sum_{t=N}^T\sum_{a=1}^A\frac{1}{\sqrt{c_{t,f_t(a)}}}\Biggr]>1-\delta.
    \end{equation*}
    Using Lemma~\ref{appendix:lemma} gives
    \begin{equation*}
        \forall\delta\in(0,1),\mathbb{P}\Biggr[R_{N:T}<2\sqrt{2ANT\ln(2NG(T,A)/\delta)}\frac{\max\mathcal{S}}{\min\mathcal{S}}\Biggr]\\>1-\delta.
    \end{equation*}
    Plugging this bound on $R_{N:T}$ into Equation~\ref{eq:regret_first_term_bounded} gives Equation~\ref{eq:min_width_regret_bound}, completing the proof.
\end{proof}

\section{Additional Experimental Results}
\label{sec:additional_results}
\definecolor{minucb}{RGB}{0, 128, 0}
\begin{table}[h]
\caption{Five algorithms are ranked in descending order by long-range performance: \textsc{Min-Width} (M-W), \textsc{Min-UCB} (M-UCB), \textsc{No-Sharing} (N-S), \textsc{CUCB}, and \textsc{UCB}.}
\centering
\begin{tabular}{|c|c|c|}
\hline
    Arm Means & Sensitivities & Long-Range Performance \\ 
\hline
 \multirow{4}{*}{0.1,0.9} & 0.1,0.9 & \color{blue}M-W\color{black}, M-UCB, N-S, UCB, CUCB \\ 
    & 0.5,0.9 & \color{blue}M-W\color{black}, M-UCB, N-S, UCB, CUCB \\
    & 0.1,0.5 & \color{blue}M-W\color{black}, M-UCB, N-S, UCB, CUCB \\
    & 0.4,0.6 & \color{blue}M-W\color{black}, M-UCB, N-S, UCB, CUCB \\
    \hline
    \multirow{4}{*}{0.5,0.9} & 0.1,0.9 & \color{blue}M-W\color{black}, M-UCB, N-S, UCB, CUCB \\ 
    & 0.5,0.9 & \color{blue}M-W\color{black}, M-UCB, N-S, UCB, CUCB \\
    & 0.1,0.5 & \color{blue}M-W\color{black}, M-UCB, N-S, UCB, CUCB \\
    & 0.4,0.6 & \color{blue}M-W\color{black}, M-UCB, N-S, UCB, CUCB \\
    \hline
    \multirow{4}{*}{0.1,0.5} & 0.1,0.9 & \color{blue}M-W\color{black}, M-UCB, N-S, UCB, CUCB \\ 
    & 0.5,0.9 & \color{blue}M-W\color{black}, M-UCB, N-S, UCB, CUCB \\
    & 0.1,0.5 & \color{blue}M-W\color{black}, M-UCB, N-S, UCB, CUCB \\
    & 0.4,0.6 & \color{blue}M-W\color{black}, M-UCB, N-S, UCB, CUCB \\
    \hline
    \multirow{4}{*}{0.4,0.6} & 0.1,0.9 & \color{blue}M-W\color{black}, M-UCB, N-S, UCB, CUCB \\ 
    & 0.5,0.9 & \color{blue}M-W\color{black}, M-UCB, N-S, UCB, CUCB \\
    & 0.1,0.5 & \color{blue}M-W\color{black}, M-UCB, N-S, UCB, CUCB \\
    & 0.4,0.6 & \color{blue}M-W\color{black}, M-UCB, N-S, UCB, CUCB  \\
    \hline
    \multirow{2}{*}{0.1,0.4,0.6} & 0.1,0.9 & \color{blue}M-W\color{black}, M-UCB, N-S, UCB, CUCB \\ 
    & 0.1,0.5,0.9 & \color{blue}M-W\color{black}, M-UCB, N-S, UCB, CUCB \\
    \hline
    \multirow{5}{*}{0.1,0.2,0.9} & 0.1,0.9 & \color{minucb}M-UCB\color{black}, N-S, M-W, UCB, CUCB \\ 
    & 0.5,0.9 & \color{minucb}M-UCB\color{black}, N-S, M-W, UCB, CUCB \\
    & 0.1,0.5 & \color{minucb}M-UCB\color{black}, N-S, M-W, UCB, CUCB \\
    & 0.4,0.6 & \color{minucb}M-UCB\color{black}, N-S, M-W, UCB, CUCB \\
    & 0.7,0.9 & \color{minucb}M-UCB\color{black}, N-S, M-W, UCB, CUCB \\
    \hline
    \multirow{5}{*}{0.1,0.5,0.9} & 0.1,0.9 & \color{minucb}M-UCB\color{black}, N-S, M-W, UCB, CUCB \\ 
    & 0.5,0.9 & \color{minucb}M-UCB\color{black}, N-S, M-W, UCB, CUCB \\
    & 0.5,0.5 & \color{cyan}CUCB\color{black}, M-UCB, N-S, M-W, UCB \\
    & 0.1,0.5,0.9 & \color{blue}M-W\color{black}, M-UCB, N-S, UCB, CUCB \\
    & 0.1,0.2,0.9 & \color{blue}M-W\color{black}, M-UCB, N-S, UCB, CUCB \\
    \hline
    \multirow{3}{*}{0.1,0.8,0.9} & 0.1,0.9 & \color{blue}M-W\color{black}, M-UCB, N-S, UCB, CUCB \\ 
    & 0.1,0.5 & \color{blue}M-W\color{black}, M-UCB, N-S, UCB, CUCB \\
    & 0.1,0.5,0.9 & \color{blue}M-W\color{black}, M-UCB, N-S, UCB, CUCB \\
    \hline
    \multirow{4}{*}{0.4,0.6,0.9} & 0.5,0.9 & \color{minucb}M-UCB\color{black}, N-S, M-W, UCB, CUCB \\ 
    & 0.4,0.6 & \color{minucb}M-UCB\color{black}, N-S, M-W, UCB, CUCB \\
    & 0.7,0.9 & \color{minucb}M-UCB\color{black}, N-S, M-W, UCB, CUCB \\
    & 0.1,0.5,0.9 & \color{blue}M-W\color{black}, M-UCB, N-S, UCB, CUCB \\
    \hline
    0.1,0.4,0.6,0.9 & 0.7,0.9 & \color{minucb}M-UCB\color{black}, N-S, M-W, UCB, CUCB
\label{table:synthetic-results}
\end{tabular}
\end{table}

A representative graph for the $2\times2$ case is shown in Figure~\ref{fig:2x2}.

\begin{figure}[H]
    \centering
    \includegraphics[width=0.5\columnwidth]{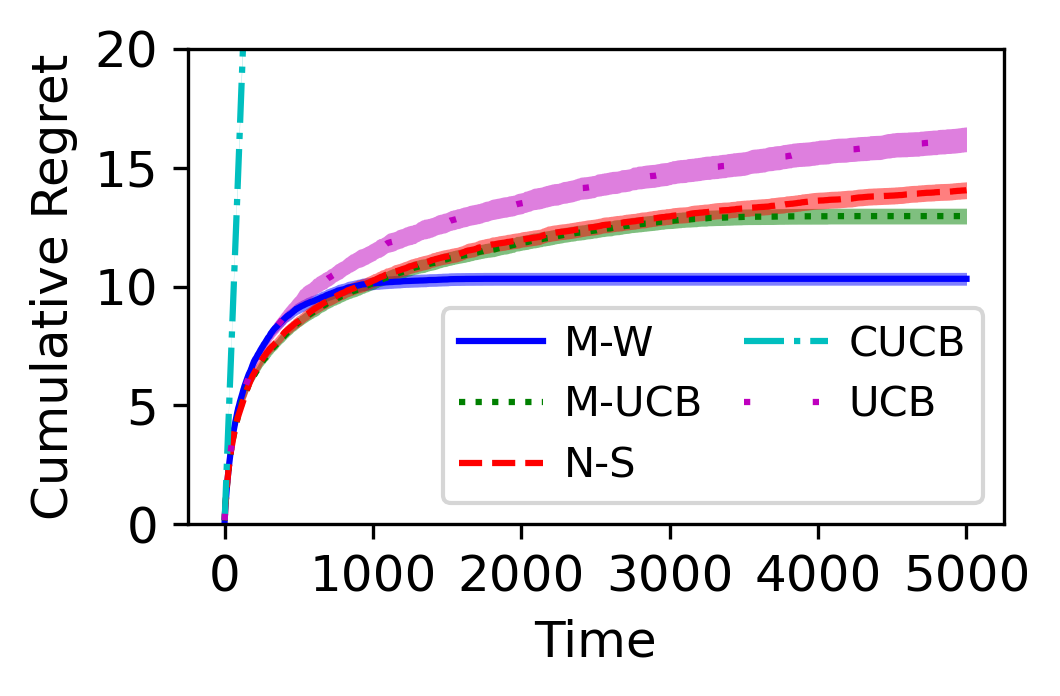}
    \caption{Regret plotted over time for \textsc{Min-Width} (M-W), \textsc{Min-UCB} (M-UCB), \textsc{No-Sharing} (N-S), \textsc{CUCB}, and \textsc{UCB} averaged over 300 trials with $\mu=\{0.1,0.5\}$ and $\mathcal{S}=\{0.1,0.9\}$.}
    \label{fig:2x2}
\end{figure}

\begin{figure}[H]
    \centering
    \includegraphics[width=0.5\columnwidth]{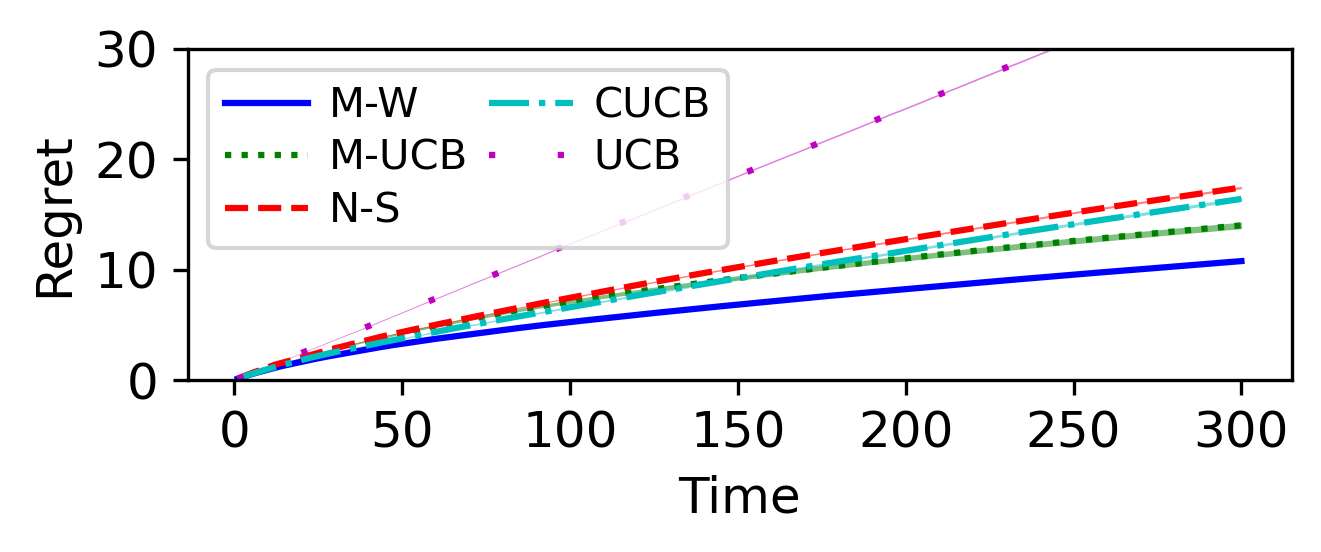}
    \caption{Regret plotted over time for \textsc{Min-Width} (M-W), \textsc{Min-UCB} (M-UCB), \textsc{No-Sharing} (N-S), \textsc{CUCB}, and \textsc{UCB} averaged over 500 trials with $\mu=\{0.05,0.1,0.12,0.15,0.25,0.3\}$, $\mathcal{S}=\{0.8,0.8,0.8,0.95,0.95\}$, $\tilde{\mathcal{S}}=\{0.85,0.85,0.85,0.98,0.98\}$.}
    \label{fig:overestimate}
\end{figure}

\begin{figure}[H]
    \centering
    \includegraphics[width=0.5\columnwidth]{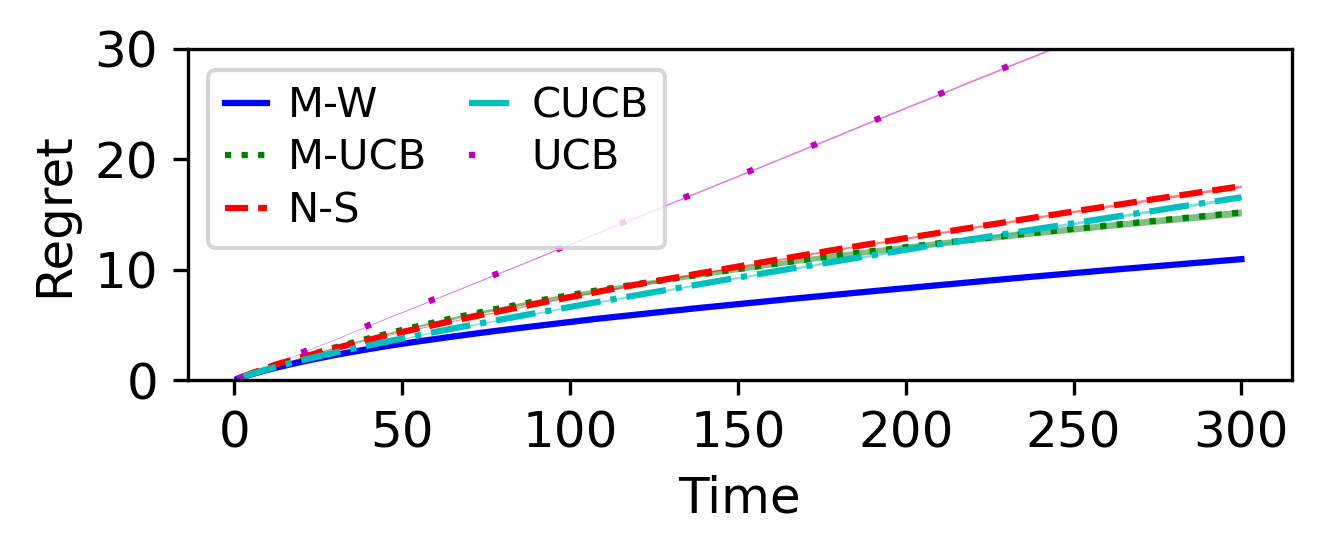}
    \caption{Regret plotted over time for \textsc{Min-Width} (M-W), \textsc{Min-UCB} (M-UCB), \textsc{No-Sharing} (N-S), \textsc{CUCB}, and \textsc{UCB} averaged over 500 trials with $\mu=\{0.05,0.1,0.12,0.15,0.25,0.3\}$, $\mathcal{S}=\{0.8,0.8,0.8,0.95,0.95\}$, $\tilde{\mathcal{S}}=\{0.75,0.75,0.75,0.9,0.9\}$.}
    \label{fig:underestimate}
\end{figure}

\begin{figure}[H]
    \centering
    \includegraphics[width=0.5\columnwidth]{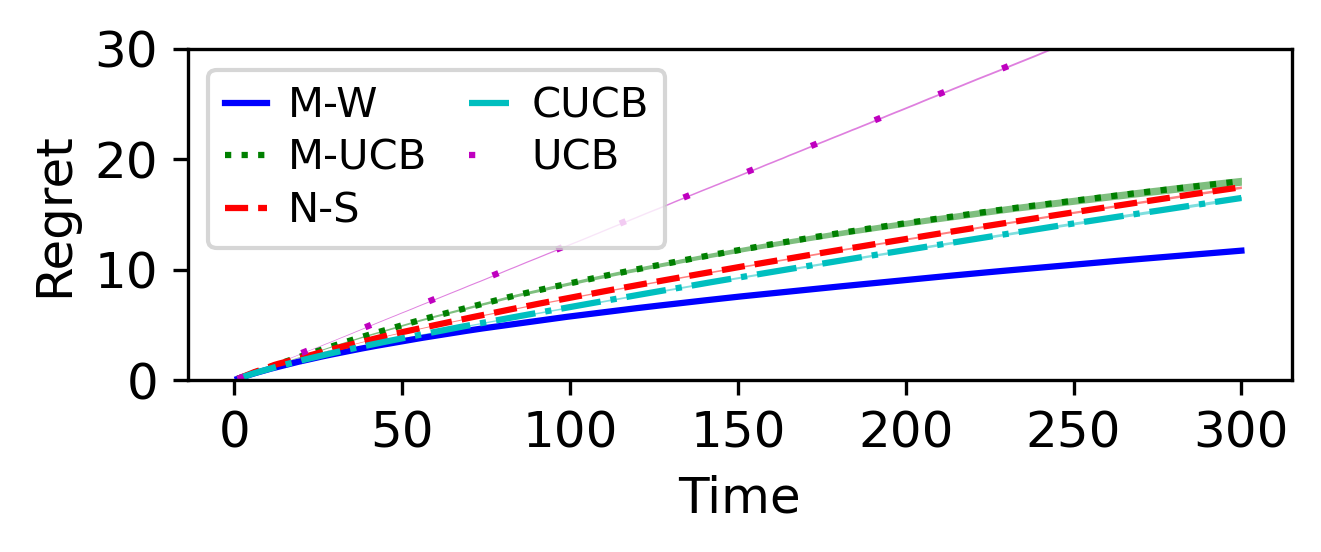}
    \caption{Regret plotted over time for \textsc{Min-Width} (M-W), \textsc{Min-UCB} (M-UCB), \textsc{No-Sharing} (N-S), \textsc{CUCB}, and \textsc{UCB} averaged over 500 trials with $\mu=\{0.05,0.1,0.12,0.15,0.25,0.3\}$, $\mathcal{S}=\{0.8,0.8,0.8,0.95,0.95\}$, $\tilde{\mathcal{S}}=\{0.75,0.75,0.75,0.98,0.98\}$.}
    \label{fig:mix}
\end{figure}